\documentclass[10pt,conference,a4paper]{IEEEtran}

\pdfpageattr{/Group << /S /Transparency /I true /CS /DeviceRGB >>}

\usepackage[utf8]{inputenc}
\usepackage[T1]{fontenc}
\usepackage{textcomp}
\usepackage{microtype}
\usepackage{calc}
\usepackage[normalem]{ulem}

\usepackage{lipsum}

\usepackage[range-phrase=--,per-mode=symbol-or-fraction,binary-units=true,range-units=single,list-units=single,detect-all,forbid-literal-units=true]{siunitx}
\AtBeginDocument{
\DeclareSIUnit\bit{bit}
\DeclareSIUnit\byte{Byte}
\DeclareSIUnit\decibelm{dBm}
\DeclareSIUnit\vehicle{veh}
}

\usepackage{silence}
\usepackage{csquotes}

\usepackage{amsmath}
\usepackage{amssymb}
\usepackage{amsfonts}
\usepackage{amsthm}

\usepackage{mathtools}

\usepackage{booktabs}
\usepackage{url}

\usepackage[inline]{enumitem}

\usepackage[caption=false,font=footnotesize]{subfig}
\usepackage[pdftex]{graphicx}
\DeclareGraphicsExtensions{.pdf,.png,.jpg}
\usepackage{tikz}
\usetikzlibrary{arrows}
\usetikzlibrary{calc}
\usetikzlibrary{chains}
\usetikzlibrary{scopes}

\usepackage[american]{babel}
\usepackage{hyphenat}

\usepackage{algorithmic}
\usepackage{algorithm}

\usepackage[capitalize,noabbrev,nameinlink]{cleveref}

\RequirePackage{xstring}
\RequirePackage{xparse}
\RequirePackage[]{acro}
\makeatletter
\@ifpackagelater{acro}{2019/02/17}{
	\NewDocumentCommand\acrodef{mO{#1}mG{}}{\DeclareAcronym{#1}{short={#2}, long={#3}, foreign-plural={}, #4}}
}{
	\NewDocumentCommand\acrodef{mO{#1}mG{}}{\DeclareAcronym{#1}{short={#2}, long={#3}, #4}}
}
\makeatother
\acrodef{AIFS}{Arbitration Inter-Frame Spacing}
\acrodef{AoI}{Age of Information}
\acrodef{AWGN}{Additive White Gaussian Noise}
\acrodef{BS}{Base Station}
\acrodef{BSS}{Basic Service Set}
\acrodef{CAM}{Cooperative Awareness packet}
\acrodef{CPM}{Collective Perception packet}
\acrodef{CBF}{Contention-Based Forwarding}
\acrodef{CBR}{Channel Busy Ratio}
\acrodef{CDF}{Cumulative Distribution Function}
\acrodef{CCDF}{Complementary Cumulative Distribution Function}
\acrodef{CDMA}{Code Division Multiple Access}
\acrodef{CCH}{Control Channel}
\acrodef{CSMA}{Carrier-Sense Multiple Access}
\acrodef{C-ITS}{Cooperative Intelligent Transportation System}{short-plural-form={C-ITS}}
\acrodef{DENM}{Decentralized Environmental Notification packet}
\acrodef{DIFS}{DCF Inter-Frame Space}
\acrodef{DSRC}{Dedicated Short-Range Communication}
\acrodef{DTMC}{Discrete Time Markov Chain}
\acrodef{FCFS}{First Come First Served}
\acrodef{GFMA}{Grant Free Multiple Access}
\acrodef{i.i.d.}{independent and identically distributed}
\acrodef{IoT}{Internet of Things}
\acrodef{ITS}{Intelligent Transportation System}{long-plural-form={}}
\acrodef{IVC}{Inter Vehicle Communication}
\acrodef{LCFS}{Last Come First Served}
\acrodef{LCFSwO}{Last Come First Served with Overwrite}
\acrodef{LDM}{Local Dynamic Map}
\acrodef{MPR}{Multi-Packet Reception}
\acrodef{mMTC}{Massive machine-type Communication}
\acrodef{M2M}{Machine-to-Machine}
\acrodef{MAC}{Medium Access Control}
\acrodef{MCM}{Maneuver Coordination packet}
\acrodef{MUSA}{Multi User Shared Access}
\acrodef{NOMA}{Non-Orthogonal Multiple Access}
\acrodef{NoB}{No Buffering}
\acrodef{OBU}{On-Board Unit}
\acrodef{OMA}{Orthogonal Multiple Access}
\acrodef{PDF}{Probability Density Function}
\acrodef{RSU}{Road Side Unit}
\acrodef{SCH}{Service Channel}
\acrodef{SCOV}{Squared Coefficient Of Variation}
\acrodef{SA}{Slotted ALOHA}
\acrodef{SIC}{Successive Interference Cancellation}
\acrodef{SNR}{Signal to Noise Ratio}
\acrodef{SNIR}{Signal to Noise plus Interference Ratio}
\acrodef{PDMA}{Pattern Division Multiple Access}
\acrodef{V2I}{Vehicle-to-Infrastructure}
\acrodef{V2V}{Vehicle-to-Vehicle}
\acrodef{V2X}{Vehicle-to-Everything}
\acrodef{VANET}{Vehicular Ad Hoc Network}
\acrodef{VLC}{Visible Light Communication}
\acrodef{UAV}{Unmanned Aerial Vehicle}

\usepackage[color]{changebar}
\def\todoCtd#1{%
	TODO: #1%
	\ifx&#1&.\fi%
	\endgroup
	\cbend
	\relax
}

\NewDocumentCommand\IEEE{ s m >{\SplitArgument{4}{/}}d[] }{%
    \IfBooleanTF{#1}{}{IEEE\,}
    \nolinebreak[2]
    #2%
    \IfNoValueTF{#3}{%
    }{%
        \sommerIEEELettersSlashed#3%
    }%
}
\newcommand{\sommerIEEELettersSlashed}[5]{%
    \IfNoValueTF{#2}{%
    }{%
        \nolinebreak[3]
    }%
	#1%
	\IfNoValueTF{#2}{}{/#2}%
	\IfNoValueTF{#3}{}{/#3}%
	\IfNoValueTF{#4}{}{/#4}%
	\IfNoValueTF{#5}{}{/#5}%
}

\newtheorem{theorem}{Theorem}

\begin{document}

\title{Asymptotic analysis of sum-rate under SIC\\
\thanks{This work was partially supported by the European Union under the Italian National Recovery and Resilience Plan (NRRP) of NextGenerationEU, partnership on Ã¢ÂÂTelecommunications of the FutureÃ¢ÂÂ (PE0000001 - program Ã¢ÂÂRESTARTÃ¢ÂÂ).}
}

\author{%
\IEEEauthorblockN{%
	Andrea Baiocchi\IEEEauthorrefmark{1}%
	,
	Asmad Razzaque\IEEEauthorrefmark{1}%
}%

\IEEEauthorblockA{
	\IEEEauthorrefmark{1}Dept.\ of Information Engineering, Electronics and Telecommunications (DIET), University of Rome Sapienza, Italy
}%


%
\texttt{%
	\{andrea.baiocchi,asmadbin.razzaque\}%
	@uniroma1.it%
}%
}

\maketitle

\begin{abstract}\nohyphens{%
Limitation of the cost of coordination and contention among a large number of nodes calls for grant-free approaches, exploiting physical layer techniques to solve collisions.
\ac{SIC} is becoming a key building block of multiple access channel receiver, in an effort to support massive \ac{IoT}.
In this paper, we explore the large-scale performance of \ac{SIC} in a theoretical framework.
A general model of a \ac{SIC} receiver is stated for a shared channel with $n$ transmitters.
The asymptotic sum-rate performance is characterized as $n \rightarrow \infty$, for a suitably scaled target \ac{SNIR}.
The probability distribution of the number of correctly decoded packets is shown to tend to a deterministic distribution asymptotically for large values of $n$.
The asymptotic analysis is carried out for any probability distribution of the wireless channel gain, assuming that the average received power level is same for all nodes, through power control. 
}\end{abstract}

\begin{IEEEkeywords}
Multiple access; SIC; large network asymptotic regime; 
\end{IEEEkeywords}

\acresetall
\IEEEpeerreviewmaketitle

%

\section{Introduction}
\label{sec:introduction}


\ac{IoT} is transforming and bringing a deep impact in every aspect of human life such as healthcare, education, Industry, agrifood chain, power grid, autonomous driving, logistic services, and smart cities \cite{fortino2017modeling,bello2019toward}. 
The evolution of \ac{IoT} toward these applications opens up a paradigm of massive access that aims at realizing efficient and reliable communication for a massive number of \ac{IoT} devices \cite{zanella2013m2m,ghavimi2014m2m}.
The main characteristics of massive access include low power, massive connectivity, short packets, and low signaling overhead, which call for more robust and flexible wireless communication networks \cite{Guo2021,Chen2021}.
Also, these requirements envision the 6G ecosystem as the playground of a data-driven society, which involves instant and unlimited connectivity to anything, ranging from tiny static sensors to autonomous objects, from anywhere, and anytime.
This paradigm calls for the design of innovative multiple access techniques to provide massive connectivity, and low latency to future massive \ac{IoT} systems.

A generic scenario of massive multiple access in \ac{IoT} involves a huge number of connected devices that give rise typically to sporadic and random traffic. 
Hence, the actual number of backlogged devices, contending for access to the medium, is largely variable.
A fundamental issue toward such a massive \ac{IoT} scenario involves reducing the cost of coordination of a varying and possibly large number of devices in an efficient and timely manner.
One promising solution toward this paradigm is the grant-free random multiple access approach, in which each active device directly transmits its data to the \ac{BS} without waiting for any permission from the access point.
Traditional grant-based access requires multiple signaling exchanges between the device and the network to allocate resources, which can be inefficient for small and infrequent data transmissions, especially in \ac{mMTC}. 
Grant-free access minimizes this overhead \cite{Jabbarvaziri2021}.

The key features of grant-free random multiple access include no prior scheduling which reduces the latency and improves the system response, scalability to support a huge number of connections, energy efficiency, eliminating the need for devices to wait for a grant, and collision management.

As for the last point, since multiple devices may attempt to transmit simultaneously, leading to collisions, grant-free random access protocols often incorporate mechanisms to handle these collisions, such as coding, advanced signal processing techniques, or \ac{NOMA} \cite{LuiPetar2018}.
\ac{NOMA} enables multiple users to share the same time-frequency resources non-orthogonally \cite{Saito2013}, thus enhancing spectral efficiency and addressing the massive connectivity demands anticipated in future wireless networks \cite{Ding2014,Yaun2021}.

Several \ac{NOMA} schemes have been proposed over the past few years, which can be divided into two main categories code-domain \ac{NOMA} and power-domain \ac{NOMA} \cite{Dai2015}.
A few variants of code-domain \ac{NOMA} includes \ac{PDMA} \cite{Chen2017}, and \ac{MUSA} \cite{Yuan2016}.
In this paper, we mainly focus on the power-domain \ac{NOMA}, as recently implemented in 3GPP LTE \cite{DingChoi2017}.
It involves the superposition of signals from different users, then exploiting \ac{SIC} to decode them \cite{Saito2013}.

Recent researches highlight a significant interest in leveraging advanced signal processing and machine learning techniques to address the challenges of \ac{GFMA} in modern wireless networks \cite{DuDonh2017,Valentini2023}.
For instance, \cite{Emir2021} proposed a novel \ac{GFMA} scheme based on deep learning techniques, achieving superior user detection accuracy and connectivity.
A reinforcement learning-based approach is introduced in  \cite{Fayaz2021} to address the challenge of massive connectivity. 
An optimal power from a pool of powers is chosen, with the help multi-agent framework, to optimize spectral efficiency and user fairness while mitigating interference.
A stochastic optimization-based adaptive grant-free scheme is proposed in \cite{ITC35}, where the transmission probability and the target \ac{SNIR} is adjusted based on the number of backlogged nodes.
It is quite evident that all these performance enhancements towards adaptive grant-free protocols come from advances carried out at the physical layer, i.e. \ac{SIC} \cite{LinDai2018,Razzaque2022}.
Hence, it becomes a crucial aspect to understand the behavior of \ac{MPR}-capable receivers, like those based on \ac{SIC}, under a massive number of users.

Several works have been done toward the design, dimensioning, and optimization of scalable \ac{MPR} systems, with respect to the number of transmitting nodes in the range of thousands.
The authors in \cite{Zanella2012} proposed a theoretical framework to investigate the effect of several parameters, such as the capture threshold and the multi-packet reception capability, with the help of \ac{SIC}, with a limited number of iterations involved in \ac{SIC} decoding process.
Enabling the \ac{SIC} mechanism brings significant performance gains in terms of throughput, especially for low capture thresholds, however, the normalized throughput increment rapidly diminishes with each iteration of \ac{SIC} \cite{Zanella2012}.
\ac{SIC} is utilized under the framework of ALOHA in  \cite{Munari2023}, with varying numbers of users driven by the duration of previous contention periods, to maximize the throughput along with optimization of the age of information.
In heavily loaded systems, the optimum \ac{SNIR} profile decreases in the \ac{SIC} decoding order \cite{Josep2017}, contrasting with uniform \ac{SNIR} profiles found in previous studies \cite{Collard2014} optimizing asymptotic capacity or under different energy constraints.

There still exists a significant gap in understanding the performance of \ac{SIC} in large-scale systems.
The role of \ac{SIC} toward massive access needs great attention, especially in view of the standardization of \ac{NOMA} in 3GPP network \cite{DingChoi2017}.

In this paper, we characterize the asymptotic behavior of sum-rate in a \ac{SIC}-based, grant-free multiple access for $n$ nodes, as $n \rightarrow \infty$
The key contributions of the paper can be summarized as follows:
\begin{itemize}
    \item We characterize the asymptotic behavior of the probability distribution of the number of correctly decoded packets utilizing \ac{SIC} in large-scale systems.
    \item We highlight the scaling of target \ac{SNIR} that guarantees the asymptotically maximum achievable sum-rate as the number of nodes grows.
    This analysis is valid for any probability distribution of the channel random gain, given that the average received power is the same for all nodes.
\end{itemize}

The paper is organized as follows.
In \cref{sec:model}, we introduce the system model and provide basic definitions.
The key results are stated and discussed in \cref{sec:analysis}, with reference to the special case of Rayleigh channel fading. 
The model is generalized in \cref{sec:discussion} to any possible distribution of channel gain. 
Conclusions are drawn in \cref{sec:conclusions}.
Proof and details of mathematical developments are given in the Appendices.

\section{System Model}
\label{sec:model}

We consider $n$ nodes sharing a common communication channel and sending packets to a sink, referred to as the \ac{BS}.
The focus is on uplink, i.e., packets are sent by nodes and addressed to the \ac{BS}.

The communication link is modeled as an \ac{AWGN} channel with block Rayleigh fading.
Let $P_j$ denote the \emph{average received} power level at the \ac{BS} from node $j$.
Let also $\mathbf{S}_j = [S_j(1),\dots,S_j(\ell)]$ the symbol string transmitted by node $j$ and $h_j$ the channel gain coefficient of node $j$.
The coefficient $h_j$ accounts for Rayleigh fading, hence $|h_j|^2$ is a negative exponential random variable.
It is normalized so that $\mathrm{E}[ |h_j|^2 ] = 1$.
The received symbol string $\mathbf{R} = [R(1),\dots,R(\ell)]$  at the \ac{BS} is modeled as
\begin{equation}
\label{eq:receivedsymbol}
\mathbf{R} = \sum_{ j = 1}^{ n }{ \sqrt{P_j} h_j \mathbf{S}_j } + \mathbf{Z}
\end{equation}
where $\mathbf{Z} = [Z(1),\dots,Z(\ell)]$ is the additive Gaussian noise, and the $Z(i)$'s are \ac{i.i.d.} Gaussian random variables, i.e., $Z(i) \sim \mathcal{CN}(0,P_N)$ for $i = 1,\dots, \ell$.
`

Assimilating interference to additive Gaussian noise, the \ac{SNIR} of node $j$ is given by
\begin{equation}
\Gamma_j = \frac{ P_j |h_j|^2 }{ P_N + I_j }
\end{equation}
where $I_j$ denotes the interference power level on the reception of packet of node $j$.
Node $j$ packet is correctly decoded if its average \ac{SNIR} at the \ac{BS} is no less than a suitable threshold $\gamma$, i.e., $\Gamma_j \ge \gamma$.
The target \ac{SNIR} $\gamma$ is tied to the spectral efficiency $\eta$ (information bit per symbol) of the communication channel according to $\eta = \log_2(1+\gamma)$, consistently with the \ac{AWGN} channel assumption.

The time axis is divided in slots.
Nodes make transmission attempts in each slot.
The slot size just fits a packet transmission time.
Let $L$ be the length of transmitted packets for all nodes and $W$ the channel bandwidth.
The time required to transmit a packet for a given target \ac{SNIR} $\gamma$ is
\begin{equation}
\label{eq:slottimeduration}
T = T(\gamma) = \frac{ L }{ W \log_2(1+\gamma) }
\end{equation}

In the following, we drop the subscript $j$ referring to the considered node to discuss transmission power setting, since the same procedure applies to each node.
We assume that node $j$ estimates the long-term average path gain of its channel to the \ac{BS}, e.g., thanks to pilot tomes transmitted periodically by the \ac{BS}.
Then, node $j$  adjusts its transmitted power level so that the average received power level at the \ac{BS} matches a prescribed value $P_0$, same for all nodes.
$P_0$ is set so that the probability of failing to decode a packet sent by a single transmitting node is no more than $\epsilon$.
Decoding of a packet sent by a single transmitting node is successful, if the \ac{SNR} exceeds $\gamma$, i.e., $P_{\text{rx}} / P_N \ge \gamma$, where $P_N$ is the background noise power level. 
Since the average received power level is $P_0$, the requirements translates to
\begin{equation}
\frac{ G_f P_0 }{ P_N } \ge \gamma \qquad \text{w.p. } 1-\epsilon
\end{equation}
Given that $G_f$ has a negative exponential \ac{PDF}, $P_0$ is to be set so that $\mathcal{P}( G_f P_0 / P_N \ge \gamma ) = e^{ - \gamma P_N/P_0} = 1-\epsilon$.
Hence, the target \ac{SNR} level $S_0$ at the receiving \ac{BS} is set as follows:
\begin{equation}
\label{eq:targetaveragerxSNR}
S_0 = \frac{ P_0 }{ P_N } = \frac{ \gamma }{ -\log(1-\epsilon) } = \frac{ \gamma }{ c }
\end{equation}
where we have introduced the constant $c = -\log(1-\epsilon)$.


We assume an ideal \ac{SIC} receiver.
Let $n$ packets be received simultaneously at the \ac{BS} in one slot.
Channel gain coefficient of node $j$ is given by $|h_j|^2$ and it is a negative exponential random variable with mean 1.
We denote the (descending) order statistics of the sequence $|h_1|,\dots,|h_n|$ with $|\tilde{h}_1|,\dots,|\tilde{h}_n|$.
By definition, we have  $|\tilde{h}_1| \ge |\tilde{h}_2| \dots \ge |\tilde{h}_n|$ (ties are broken at random).
\ac{SIC} works as follows.
Provided decoding of packets $1,\dots,k-1$ be successful, packet $k$ is decoded successfully if and only if the following inequality holds\footnote{In the following, for ease of notation, it is mean that $\sum_{ i = a }^{ b }{} = 0$ and $\prod_{ i = a }^{ b }{} = 1$, if $a > b$.}:
\begin{equation}
	\label{eq:SNIRdefinition}
	\Gamma_j = \frac{ P_0 |\tilde{h}_k|^2 }{ P_N + \sum_{ r = k+1 }^{ n }{ P_0 |\tilde{h}_r|^2 } + \xi \sum_{ r = 1 }^{ k-1 }{ P_0 |\tilde{h}_r|^2 } } \ge \gamma
\end{equation}
where $\xi$ accounts for the fraction of residual interference due to imperfect cancellation.
Perfect \ac{SIC} is obtained by setting $\xi = 0$.
On the contrary, $\xi = 1$ corresponds to a receiver not endowed with \ac{SIC} capability, where multi-packet reception is possible only thanks to capture effect\footnote{In this special case, ordering is not relevant and decoding of individual packets is done independently of one another.}.
Using \cref{eq:targetaveragerxSNR}, \cref{eq:SNIRdefinition} can be re-written as follows:
\begin{equation}
	\Gamma_j = \frac{  |\tilde{h}_k|^2 }{ c/\gamma + \sum_{ r = k+1 }^{ n }{ |\tilde{h}_r|^2 } + \xi \sum_{ r = 1 }^{ k-1 }{ |\tilde{h}_r|^2 } } \ge \gamma
\end{equation}

To simplify notation, we let $Y_j = |h_j|^2$ and $Y_{(j)} = | \tilde{h}_j |^2$ for $j = 1,\dots,n$ in the following.
The $Y_j$'s are \ac{i.i.d.} negative exponential random variables with mean 1, while the $Y_{(j)}$'s form the associate descending order statistic.
Successful decoding of packet carried by the $m$-th strongest signal with \ac{SIC} requires that $\Gamma_j \ge \gamma$ for all $j = 1,\dots,m$.
Using the notation introduced so far, the condition is stated as
\begin{equation}
	\label{eq:successfulSIC}
	\Gamma_j = \frac{  Y_{(j)} }{ c/\gamma + \sum_{ r = j+1 }^{ n }{ Y_{(r)} } + \xi \sum_{ r = 1 }^{ j-1 }{ Y_{(r)} } } \ge \gamma
\end{equation}
for all $j = 1,\dots,m$.
In \cref{eq:successfulSIC} the residual interference coefficient $\xi$ is assumed to satisfy $0 \le \xi < 1$.

If \ac{SIC} is not used at the receiver, successful decoding is still possible.
Each packet is decoded independently of all others, so that ordering does not matter.
The packet carried by the $j$-th signal is successfully decoded if
\begin{equation}
	\label{eq:successfulnoSIC}
	\frac{  Y_j }{ c/\gamma + \sum_{ r = 1, r \ne j }^{ n }{ Y_r } } \ge \gamma
\end{equation}

We introduce the notation $S_j = S_0 Y_j$, $j = 1,\dots,n$, to denote the \ac{SNR} of node $j$, where $S_)$ is the average \ac{SNR} level at the receiver (see \cref{eq:targetaveragerxSNR}), while $Y_j$ accounts for the residual randomness of the communication channel.

In next Section we analyze the asymptotic regime for $n \rightarrow \infty$ of \ac{SIC} reception, under the assumption of Rayleigh fading, i.e., that $Y_j$ be \ac{i.i.d.} negative esponential random variable with mean 1.
In \cref{sec:discussion} this assumption is relaxed, by allowing the $Y_j$'s to have any probability distribution.

%

\section{Scaling of multiple access parameters with the number of nodes}
\label{sec:scaling}

Let us consider a multiple access channel with slotted time axis and a configurable target \ac{SNIR} $\gamma$, as described in the previous Section.
Nodes that are backlogged at the beginning of a slot transmit in that slot with probability $p$, or defer a new attempt in next slot with probability $1-p$.
In this Section we motivate why the target \ac{SNIR} $\gamma$ and the transmission probability $p$ must scale with the number of backlogged nodes $n$ for stability.

To measure the efficiency of the communication channel we focus on the sum-rate metric.
The sum-rate gives the achieved spectral efficiency of the multiple access channel in bit/s/Hz \cite{Razzaque2022}.
Let us assume that a backlogged node transmits in a slot with probability $p$ and sets the target \ac{SNIR} to $\gamma$ when transmitting.
Assuming an AWGN channel over a bandwidth $W$ and a time slot size $T$, the number of bits delivered with a single successful transmission is $T \, W \, \log_2(1+\gamma)$
Let also $M_n(p,\gamma)$ denote the mean number of  successfully decoded packets, when $n$ nodes are backlogged, transmit with probability $p$, use a target \ac{SNIR} $\gamma$.
Then, the achieved sum-rate with $n$ backlogged ndoes is
\begin{equation}
\label{ }
U = U_n(p,\gamma) = \log_2(1+\gamma) M_n(p,\gamma)
\end{equation}
We have 
\begin{equation}
\label{eq:sumrateforSA}
M_n(p,\gamma) =  \sum_{ k = 1 }^{ n }{ m_k(\gamma) \binom{n}{k} p^k (1-p)^{n-k} }
\end{equation}
where we let $m_k(\gamma)$ denote the mean number of correctly decoded packets in a slot with target \ac{SNIR} equal to $\gamma$, given that $k$ nodes transmit in that slot.
In Appendix \ref{app:Appboundsonmkgamma} it is proved that the following bounds hold on $m_k(\gamma)$:
\begin{equation}
\label{eq:boundsofmkappaofgamma}
\frac{ k ( 1 - \epsilon ) }{ (1+\gamma)^{k-1} } \le m_k(\gamma) \le \frac{ k^2 ( 1 - \epsilon ) }{ (1+\gamma)^{k-1} } 
\end{equation}
for all $k \ge 0$ and all $\gamma > 0$.

Exploiting those bounds, in Appendix \ref{app:AppUfixedparameters} it is shown that $U_n(p,\gamma) \rightarrow 0$ as $n \rightarrow \infty$, for any fixed values of $p$ and $\gamma$.
We must therefore consider adaptive values of $p_n$ and $\gamma_n$.

In Appendix \ref{app:scalingofpandgamma} it is shown that the sum-rate vanishes as $n$ grow, if we let $p$ and $\gamma$ scale with $n$ in such a way that $p = p_n \sim O(1/n^a)$ and $\gamma = \gamma_n \sim O(1/n^b)$ with $a$ and $b$ non-negative numbers such that $a+b \ne 1$.
Putting this result together with that of Appendix \ref{app:AppUfixedparameters}, we conclude that the only way to get a non vanishing limit of the sum-rate as the number of backlogged nodes grows can possibly be adapting $p_n$ and $\gamma_n$ as described above, with $a+b = 1$.

Scaling the transmission probability with the number of backlogged nodes is a well known approach to the stabilitzation of Slotted ALOHA \cite{Bertsekas1992} and \ac{CSMA} \cite{Bianchi2000}.
We therefore focus on the less studied scaling, where $\gamma$ is set inversely proportional to $n$, while leaving all backlogged nodes transmit in a slot, i.e., we consider the scaling $\gamma_n = 1/(\alpha n)$ and $p_n = 1$.

%

\section{Asymptotic Analysis}
\label{sec:analysis}

The first step is to transform the inequalities in \cref{eq:successfulSIC} into a simpler form.
This is the purpose of next theorem.

\begin{theorem}
\label{theo:introVj}
Let $V_j, \, j = 1,\dots,n$, be a non-negative random variables defined by 
\begin{equation}
V_j = \sum_{ k = j }^{ n }{ b_{jk} X_k }
\end{equation}
where $X_k, \, k = 1,\dots,n$, are \ac{i.i.d.} negative exponential random variables with mean 1 and
\begin{equation}
\label{eq:coeffbjkdef}
b_{jk} = \begin{cases}
     - \gamma \xi & k = 1,\dots,j-1, \\
     \frac{ 1 + j \gamma - (j-1) \gamma \xi }{ k } - \gamma & k = j,\dots,n.
\end{cases}
\end{equation}
for $j = 1,\dots,n$.
The packet carried by the $m$-th strongest signal is decoded successfully if and only if $V_j \ge c$ for $j = 1,\dots,m$.
\end{theorem}

\begin{proof}
See Appendix A.
\end{proof}

The random variables $V_j, \, j = 1,\dots,n$ are not independent of one another.
The joint condition for successful decoding stated in \cref{theo:introVj} is not easy to evaluate.

In case $\xi = 0$, the probability $p_V(j) = \mathcal{P}( V_j \ge c )$, can be evaluated starting from the Laplace transform of the \ac{PDF} of $V_j$
\begin{equation}
\varphi_{V_j}(s) = \mathrm{E}[ e^{ - s V_j } ] = \frac{ 1 }{ \prod_{ k = j }^{ n }{ ( 1 + b_{jk} s ) } }
\end{equation}

This transform can be inverted in closed form, but the numerical evaluation of the obtained formulas is highly unstable, due to sums of alternating sign terms.
Numerical error propagation forbids a direct evaluation of the probabilities $p_V(j)$.
Numerical experience suggests that it is much more reliable and efficient to use numerical inversion of the Laplace transform directly, e.g., exploiting the Fourier algorithm in \cite{abate1995numerical}.

\cref{fig:PVjgec} illustrates $p_V(j) = \mathcal{P}( V_j \ge c )$ as a function of $j$ for four values of the number of nodes $n$, with $\xi = 0$ and $\gamma = 1/( \alpha n )$ for a given constant $\alpha > 0$ ($\alpha = 0.32$ in the left plot, $\alpha = 0.38$ in the right one).
The realizations of the random variable $V$ on the $x$ axis are normalized as $j/n$.
The number of nodes is stretched to possibly unrealistically large values for two reasons.
A first reason is to highlight the asymptotic behavior of the probabilities $p_V(j)$ as $n$ gets large.
The second reason derives from a system point of view: massive \ac{IoT} calls for the possibility of even very large set of concurrently transmitting nodes.
Even if this may not be feasible with today's technology, it is interesting to understand what happens in such a regime.

\begin{figure}
	\centering
	\subfloat[$\alpha = 0.32$]{ \label{fig:PVjgecalfa032}
	    \includegraphics[width=.45\columnwidth]{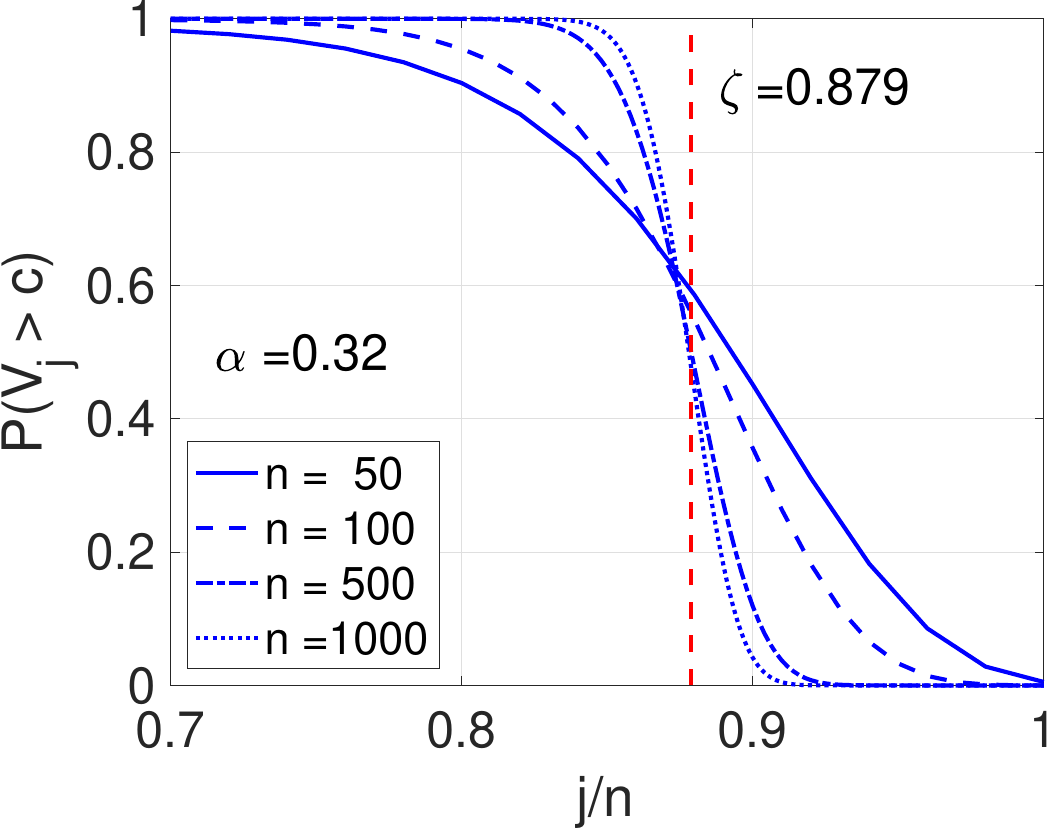} } \;%
	\subfloat[$\alpha = 0.38$]{ \label{fig:PVjgecalfa038} 
	    \includegraphics[width=.45\columnwidth]{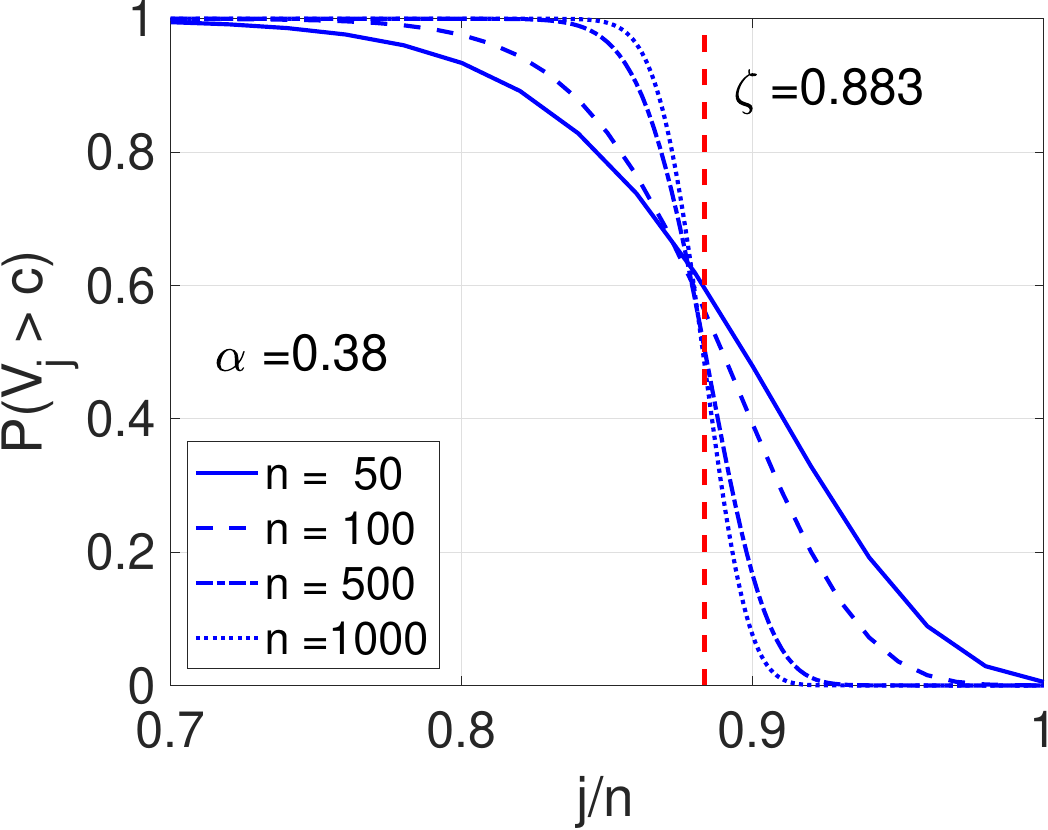} }
	\caption{Probability that the random variable $V_j$ exceeds the threshold $c$ for $j = 1,\dots,n$. The abscissa is normalized as $j/n$. The required \ac{SNIR} threshold is set to $\gamma = 1/( \alpha n )$, with $\alpha = 0.32$ in the plot on the left, $\alpha = 0.38$ in the plot on the right.}
	\label{fig:PVjgec}
\end{figure}

The probabilities $p_V(j)$ form a decreasing sequence with $j$, which matches intuition, given that higher values of $j$ are related to a higher number of successfully decoded packets.
Hence, $p_V(j) \approx 1$ for lower values of $j$, while $p_V(j)$ becomes negligible for $j$ approaching $n$.
In fact, for $j = n$ it is easy to verify that $p_V(n) = e^{-nc} = (1-\epsilon)^n$.
The transition from 1 to 0 as $j$ is increased is apparent in \cref{fig:PVjgec}.
The striking result is that the transition gets steeper as $n$ grows, i.e., \ac{SIC} decoding in our model moves towards a phase transition phenomenon as the number of concurrent transmissions gets large.
The convergence is rather slow, since very large values of $n$ have to be considered to appreciate a sharp transition.
The transition corresponds to a decoding threshold, highlighted as $\zeta$ in the plots. 
Up to a fraction $\zeta$ of the $n$ transmissions are decoded successfully as $n \rightarrow \infty$.
In terms of random variables,  \cref{fig:PVjgec} tells us that the random variables $V_j$ tend to become deterministic quantities.
This is formally proved in the following theorem.

\begin{theorem}
\label{theo:asymptoticpropertiesofmomentsofVj}
Let $\mu_j(n)$ and $\sigma^2_j(n)$ denote the mean and the  variance of $V_j$ respectively, for $n$ concurrent transmissions and for $j = 1,\dots,n$.
Let $\gamma = 1/(\alpha n)$ for a positive constant $\alpha$.
Let also
\begin{equation}
\label{eq:definitionfofx}
f_{\alpha,\xi}(x) = - \left( 1 + \frac{ (1-\xi) x }{ \alpha } \right) \log x - \frac{ 1 - (1-\xi) x }{ \alpha }
\end{equation}
for $x \in (0,1]$.
Then
\begin{align}
&\lim_{ n \rightarrow \infty }{ \frac{ 1 }{ n } \sum_{ j = 1 }^{ n }{ \left| \mu_j(n) - f_{\alpha,\xi}\left( \frac{ j }{ n } \right) \right| } } = 0  \\
&\lim_{ n \rightarrow \infty }{ \frac{ \sigma_j(n) }{ \mu_j(n) } } = 0
\end{align}
Hence, for any $\epsilon > 0$, we have
\begin{equation}
\lim_{ n \rightarrow \infty }{ \mathcal{P}\left( \left| \frac{ V_j }{ \mu_j(n) } - 1 \right| > \epsilon \right) } = 0 
\end{equation}
\end{theorem}

\begin{proof}
See Appendix B.
\end{proof}

Note that the function $f_{\alpha,\xi}(x)$ can be written as
\begin{equation}
\label{eq:relationshiofandg}
f_{\alpha,\xi}(x) = - \frac{ 1 }{ \alpha } + g_{ \frac{ \alpha }{ 1 - \xi } }(x)
\end{equation}
where
\begin{equation}
\label{eq:definitionofgbetaofx}
g_\beta(x) = - \left( 1 + \frac{ x }{ \beta } \right) \log x + \frac{ x }{ \beta }
\end{equation}
This representation of $f_{\alpha,\xi}(x$ shows that, apart from a constant, it depends only on the parameter $\beta = \alpha/(1-\xi)$.

\cref{theo:asymptoticpropertiesofmomentsofVj} says that the mean values of the random variables $V_j$ tend to lie on the curve $f_{\alpha,\xi}(x)$, with $x = j/n$, as $n$ grows, while the coefficient of variation (the ratio of the standard deviation to the mean) shrinks to 0.
Hence, for each $j$ with $1 \le j \le n$, $V_j$ tends to a deterministic random variable that takes the value $\mu_j(n)$ as $n \rightarrow \infty$.
This asymptotic regime sets on because of the crucial condition on the required \ac{SNIR} threshold, namely, $\gamma = 1/( \alpha n )$.
Intuitively, if the required \ac{SNIR} scales with the number $n$ of concurrent transmissions, then the mean of $V_j$ tend to a positive value.
The implications on the number of correctly decodable packets are addressed by next theorem.
Before stating it, numerical evidence of the convergence of the random variables $V_j$ to a deterministic quantity is given in curve plots, from which we appreciate \emph{how} the asymptotes are achieved.

The sequence of means $\mu_j(n)$ (discrete stems ending with a circle) is compared with $f_{\alpha,\xi}(x)$ (thick red line) in \cref{fig:meanofVjcomparedtofofx} for $\alpha = 0.32$, $\xi = 0$ and several values of $n$, by placing the mean values at points $j/n, \, j = 1,\dots,n$.
It is apparent that mean values match more and more closely as $n$ grows.
It is also evident that convergence is rather slow.
It appears to be very accurate for $n = 500$, but it is still not perfect for $n = 50$.

\begin{figure}
	\centering
	\subfloat[$n = 50$]{ \label{fig:meanofVjcomparedtofofx50}
	    \includegraphics[width=.45\columnwidth]{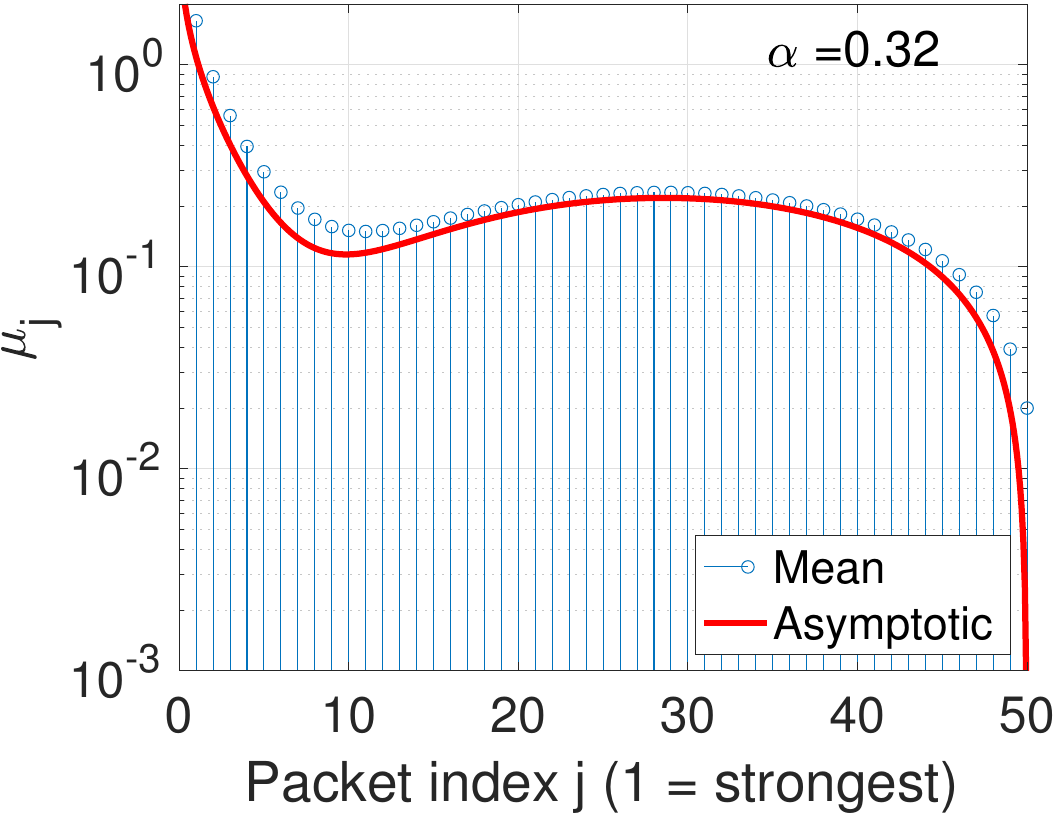} } \;%
	\subfloat[$n = 100$]{ \label{fig:meanofVjcomparedtofofx100} 
	    \includegraphics[width=.45\columnwidth]{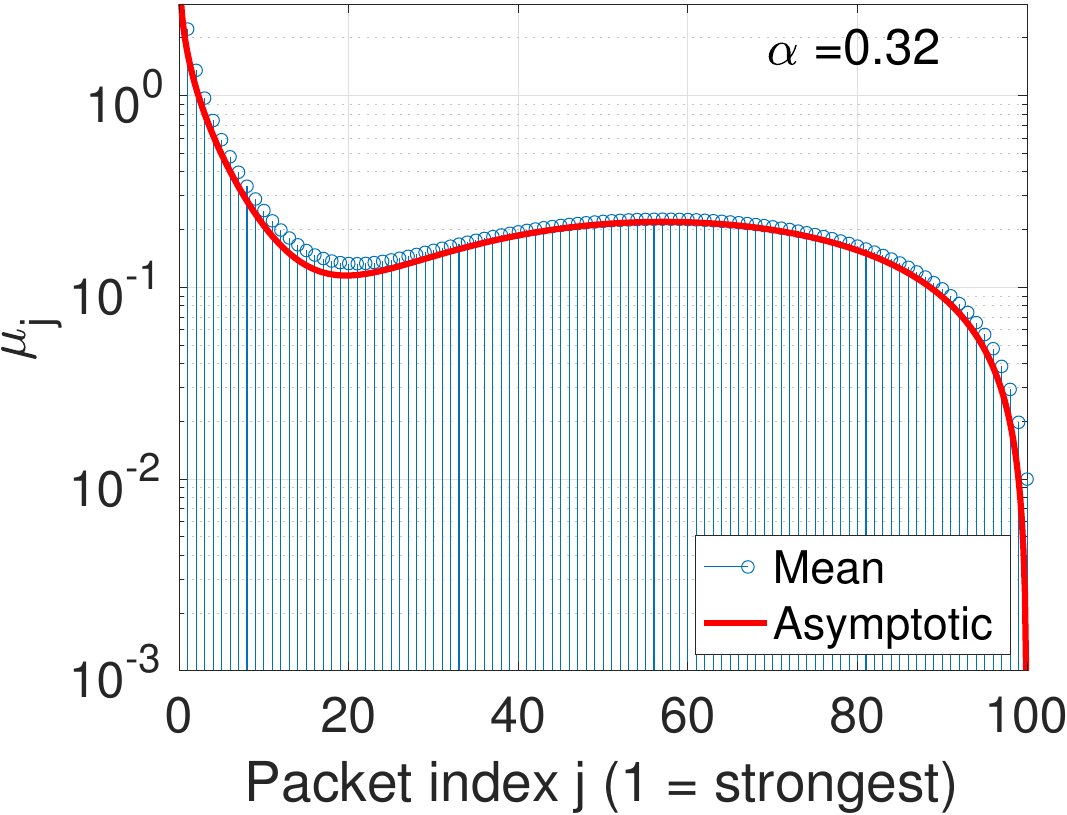} } \\
	\subfloat[$n = 500$]{ \label{fig:meanofVjcomparedtofofx500}
	    \includegraphics[width=.45\columnwidth]{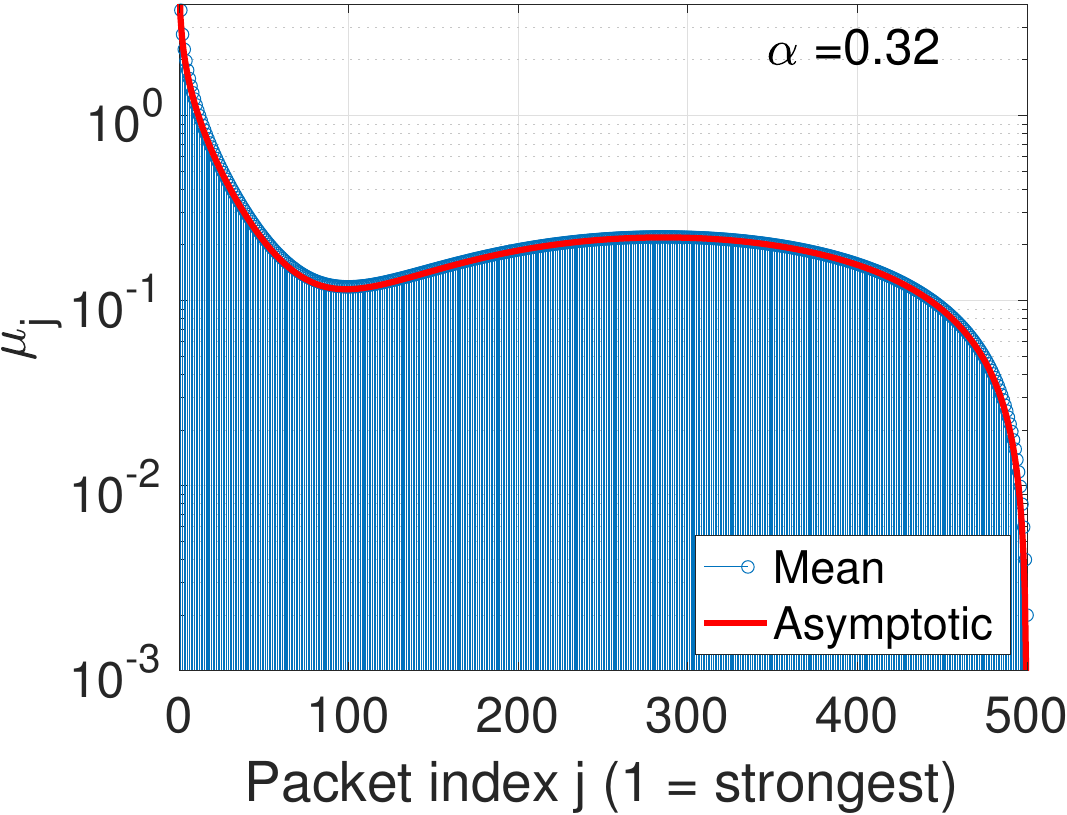} } \;%
	\subfloat[$n = 1000$]{ \label{fig:meanofVjcomparedtofofx1000} 
	    \includegraphics[width=.45\columnwidth]{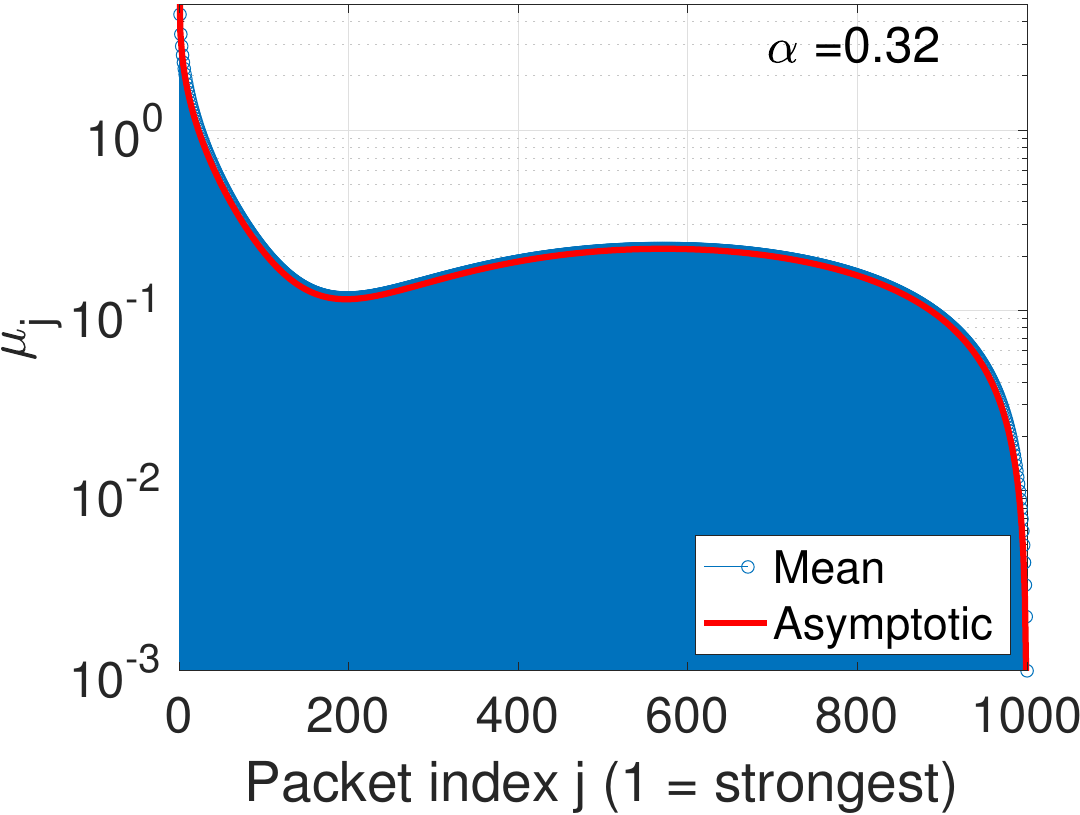} }
	\caption{Sequence of mean values $\mu_j(n)$ (discrete stems ending with a circle) compared with function $f_{\alpha,\xi}(x)$ (thick red line) as a function of $x = j/n$ for several values of $n$ and $\xi = 0$ and $\alpha = 0.32$.}
	\label{fig:meanofVjcomparedtofofx}
\end{figure}

 \cref{fig:convergenceofVjlinalfa032,fig:convergenceofVjlinalfa038} show the sequence of means $\mu_j(n)$ as a function of $x = j/n$ (solid line), as well as the asymptotic curve $f_{\alpha,\xi}(x)$ towards which the mean values approach as $n$ grows.
 The shaded area corresponds to one standard deviation from the mean, i.e., it is the area within the limits $\mu_j(n) \pm \sigma_j(n)$.
Plots are shown for several values of $n$,  for two values of $\alpha$ and for $\xi = 0$.
Studying the function $f_{\alpha,\xi}(x)$ for $x \in (0,1]$, it can be verified that its derivatives has two zeros, $x_1$ and $x_2$, with $x_1 < 1/e < x_2$, if and only if $\alpha/(1-\xi) < 1/e \approx 0.37$. 
Otherwise, the function $f_{\alpha,\xi}(x)$ is monotonously decreasing.

More in depth, the derivative of $f^\prime_{\alpha,\xi}(x)$ is $ = -1/x - [ (1-\xi) / \alpha ] \log x$.
Hence $x_1$ and $x_2$ are the roots of
\begin{equation}
\frac{ \alpha }{ 1-\xi } + x \log x = 0
\end{equation}
for $x \in (0,1]$.
These roots, when they exist, depend only on the parameter $\beta = \alpha/(1-\xi)$.
When existing, the smaller zero $x_1$ corresponds to a local minimum, while the larger $x_2$ to a local maximum.
Hence, in \cref{fig:convergenceofVjlinalfa032,fig:convergenceofVjlinalfa038} two values of $\alpha$ across the critical value $1/e$ have been picked.
The zero $x_1$ is a function of $\beta$, defined implicitly by\footnote{The function $x_1(\beta)$ is connected to Lambert's function, but we will not pursue this connection further here.}  $\beta + x_1(\beta) \log x_1(\beta) = 0$. 
Deriving both sides, it follows that $x_1^\prime(\beta) = -1/(1+\log x_1(\beta) )$, for $\beta \in (0,1/e)$.
Since $x_1(\beta) < 1/e$, it must be $x_1^\prime(\beta) > 0$ for $\beta \in (0,1/e)$, i.e., $x_1(\beta)$ is an increasing function of $\beta$, its maximum being attained for $\beta = 1/e$, where it is $x_1 = x_2 = 1/e$.
We note that the smallest positive root of $f_{\alpha,\xi}(x) = c$ is a continuous function of $\beta$ if and only if $f_{\alpha,\xi}(x_1(1/e)) = f_{\alpha,\xi}(1/e) \le c$.
Since for $\beta = 1/e$ we have $f_{\alpha,\xi}(1/e) = 3-e/(1-\xi)$, the condition becomes $\xi \ge 1 - e/(3-c)$.
For $\epsilon = 0.1$, it is $c \approx 0.1054$, hence the condition is $\xi \ge 0.0609$.
For smaller values of $\xi$ the smallest positive root of $f_{\alpha,\xi}(x) = c$ as a function of $\alpha$ has a jump. 

The plots in these figures highlight the convergence of the random variables $V_j$ to deterministic quantities, as the shaded area shrinks for increasing $n$.
It is apparent that convergence is not uniform.
The range of $x$ around the local minimum of the curve $f_{\alpha,\xi}(x)$ appears to be the most critical region for convergence.

\begin{figure}
	\centering
	\subfloat[$n = 50$]{ \label{fig:mupmsigman50lin1}
	    \includegraphics[width=.45\columnwidth]{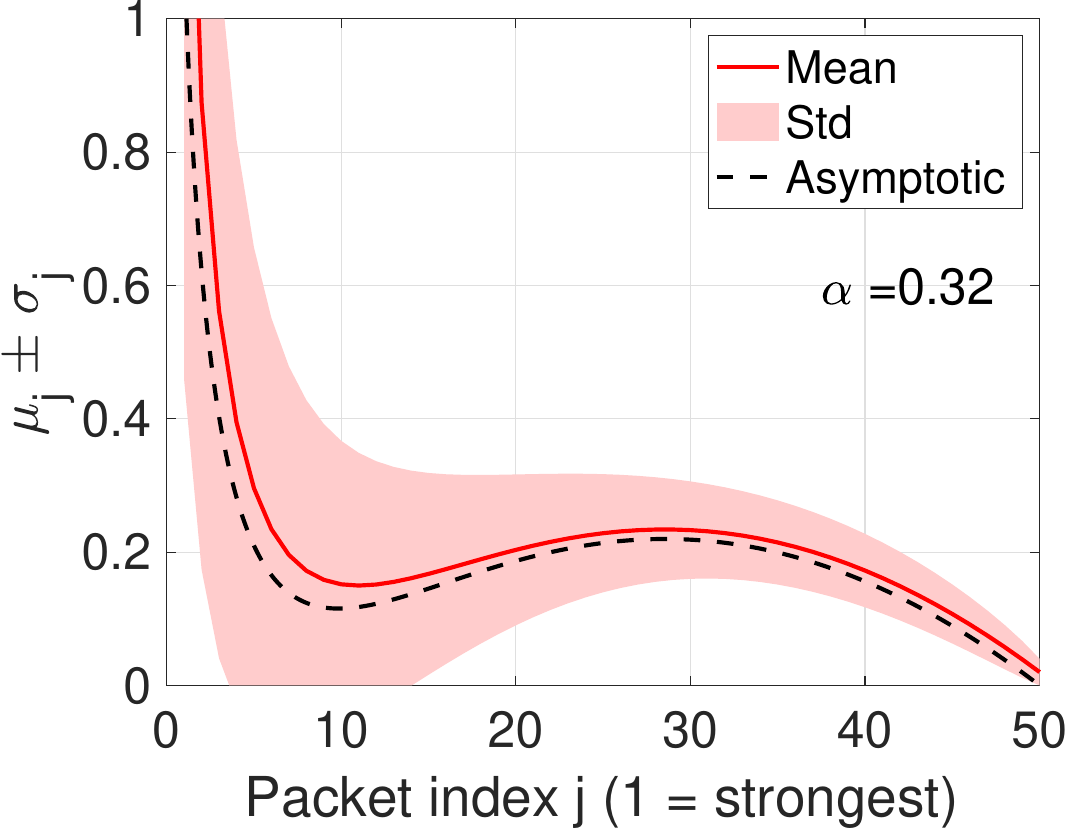} } \;%
	\subfloat[$n = 100$]{ \label{fig:mupmsigman100lin1} 
	    \includegraphics[width=.45\columnwidth]{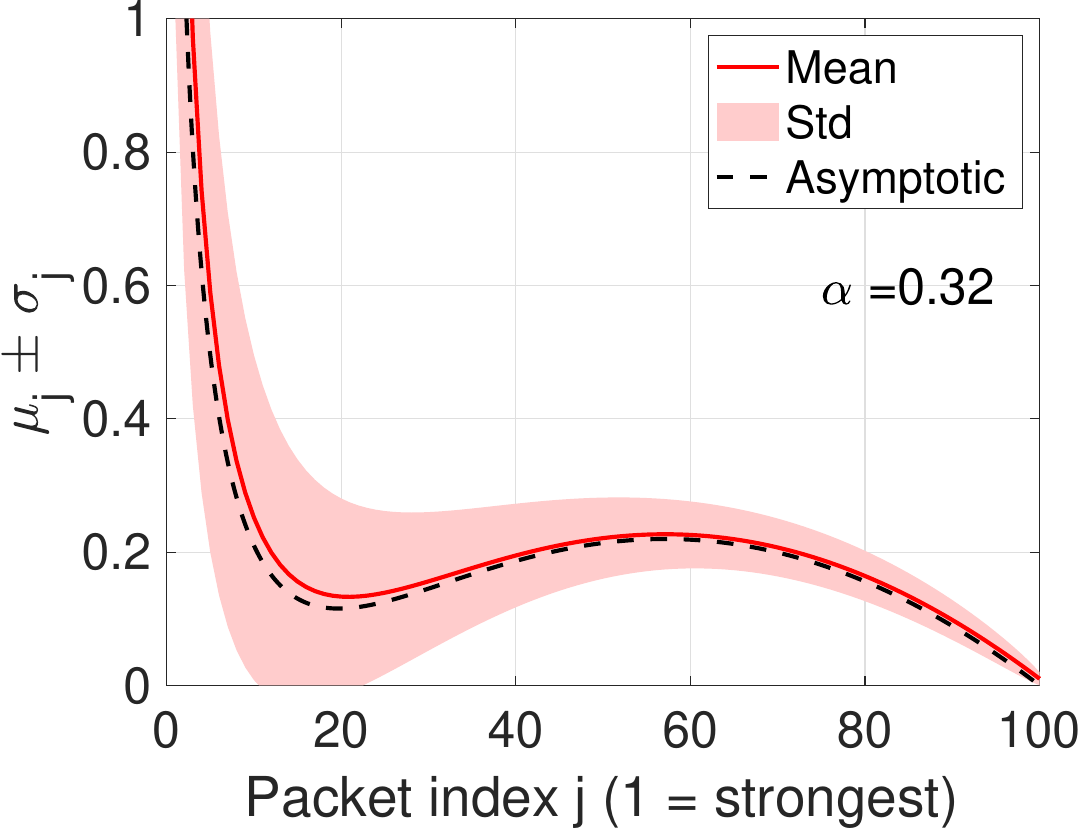} } \\
	\subfloat[$n = 500$]{ \label{fig:mupmsigman500lin1}
	    \includegraphics[width=.45\columnwidth]{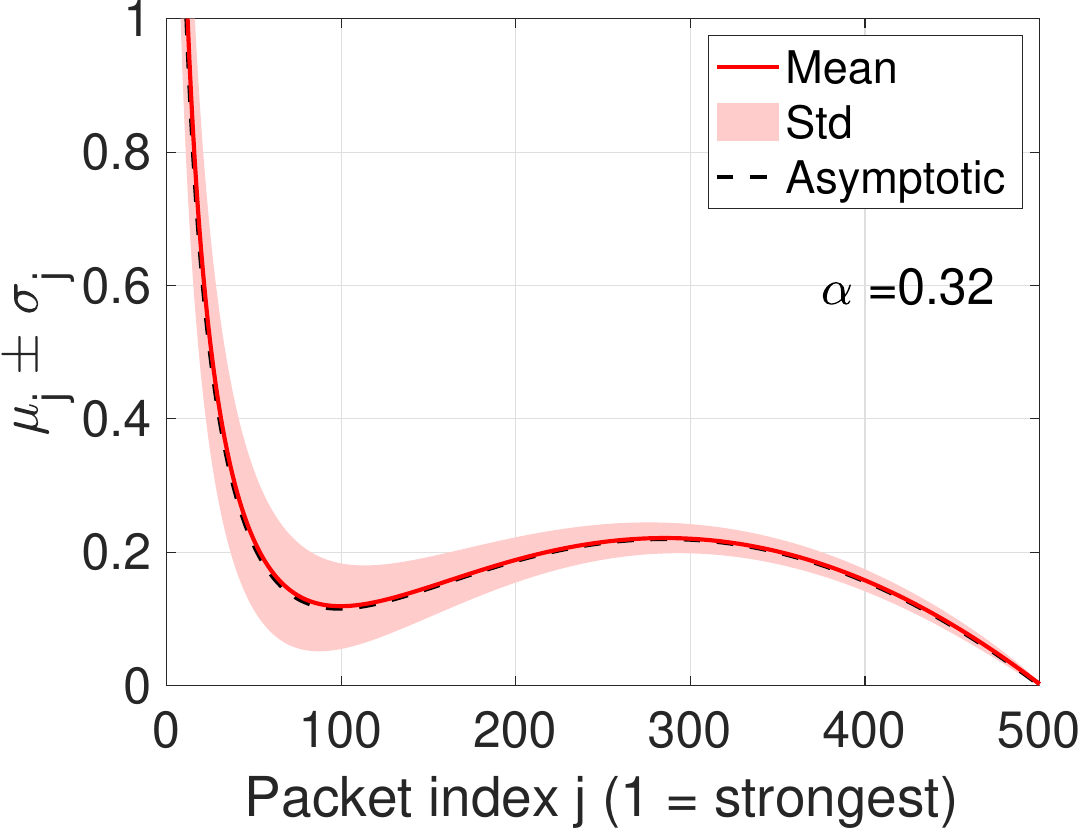} } \;%
	\subfloat[$n = 1000$]{ \label{fig:mupmsigman1000lin1} 
	    \includegraphics[width=.45\columnwidth]{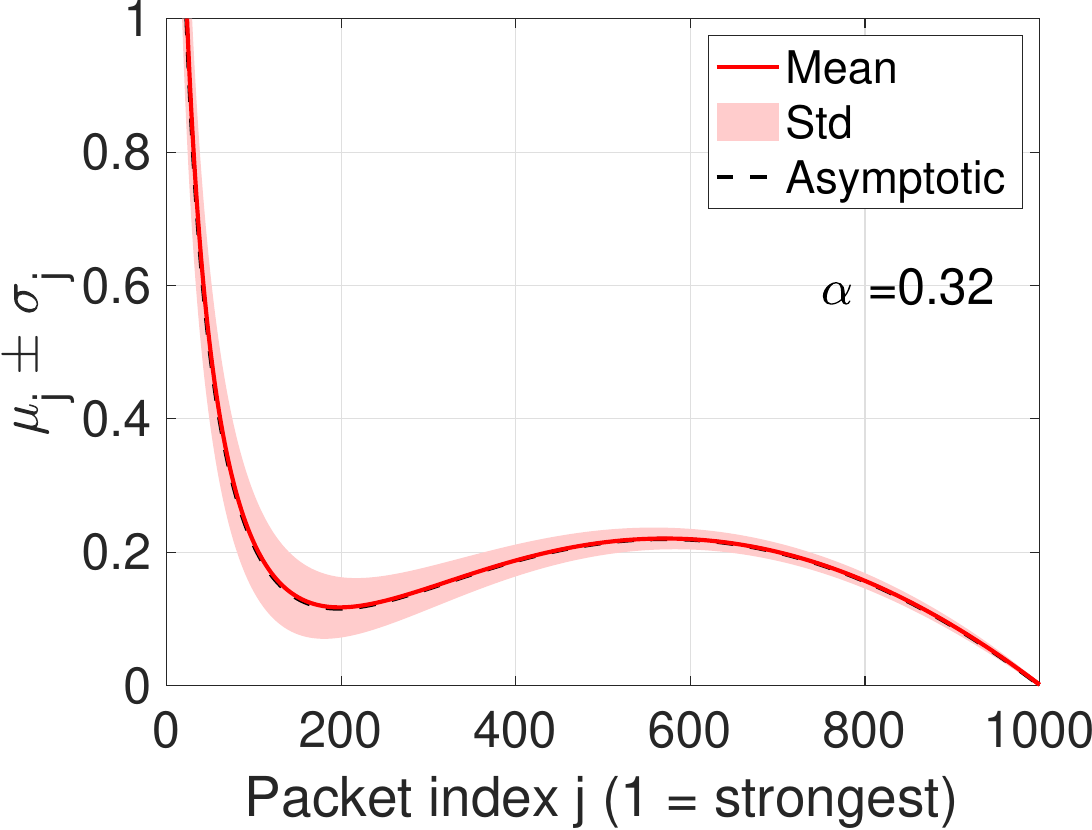} }
	\caption{Sequence of mean values $\mu_j(n)$ as a function of $x = j/n$ for several values of $n$ and $\xi = 0$. The shaded area corresponds to one standard deviation, i.e., $\mu_j(n) \pm \sigma_j(n)$. The dashed line is the function $f_{\alpha,\xi}(x)$ defined in \cref{eq:definitionfofx}. The required \ac{SNIR} threshold is set to $\gamma = 1/( \alpha n )$ with $\alpha = 0.32$}
	\label{fig:convergenceofVjlinalfa032}
\end{figure}

\begin{figure}
	\centering
	\subfloat[$n = 50$]{ \label{fig:mupmsigman50lin2}
	    \includegraphics[width=.45\columnwidth]{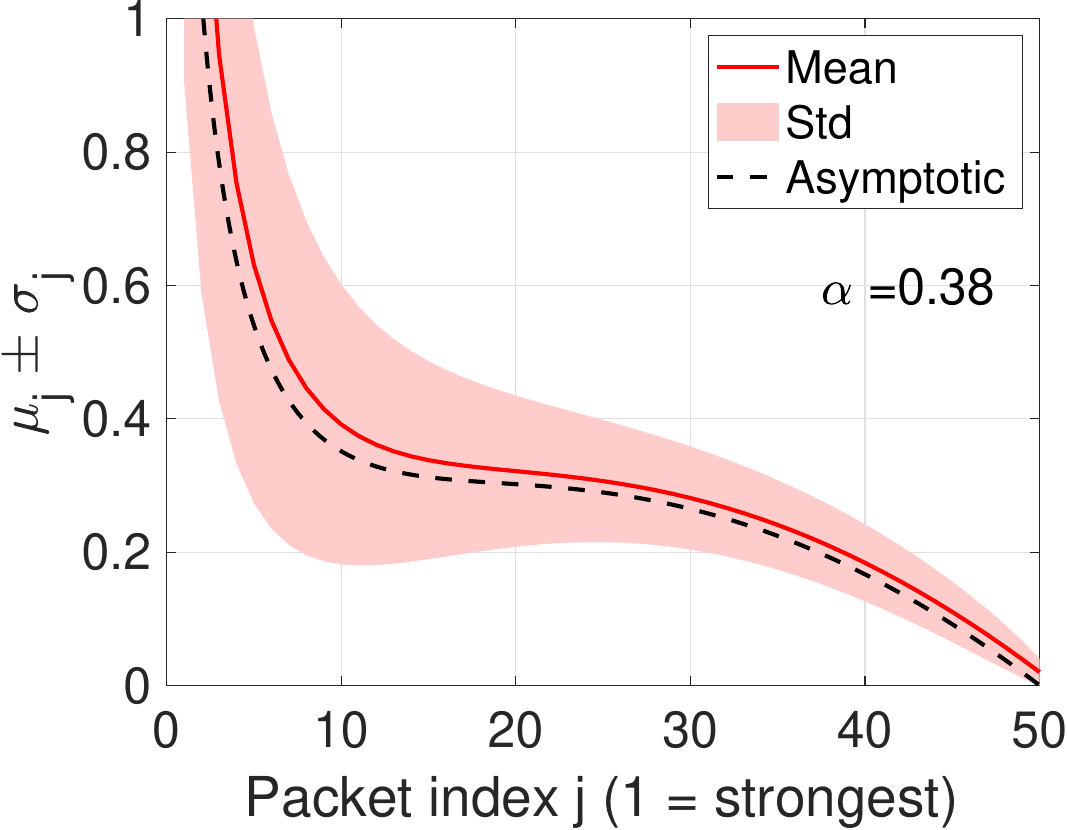} } \;%
	\subfloat[$n = 100$]{ \label{fig:mupmsigman100lin2} 
	    \includegraphics[width=.45\columnwidth]{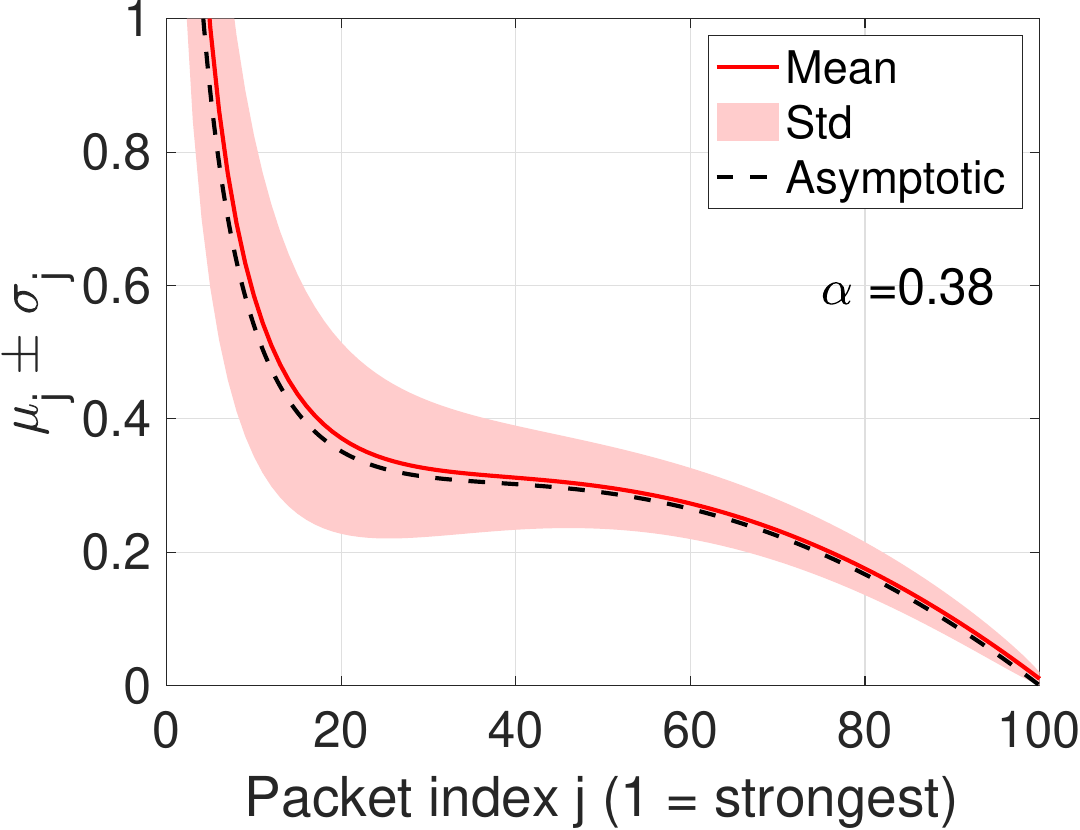} } \\
	\subfloat[$n = 500$]{ \label{fig:mupmsigman500lin2}
	    \includegraphics[width=.45\columnwidth]{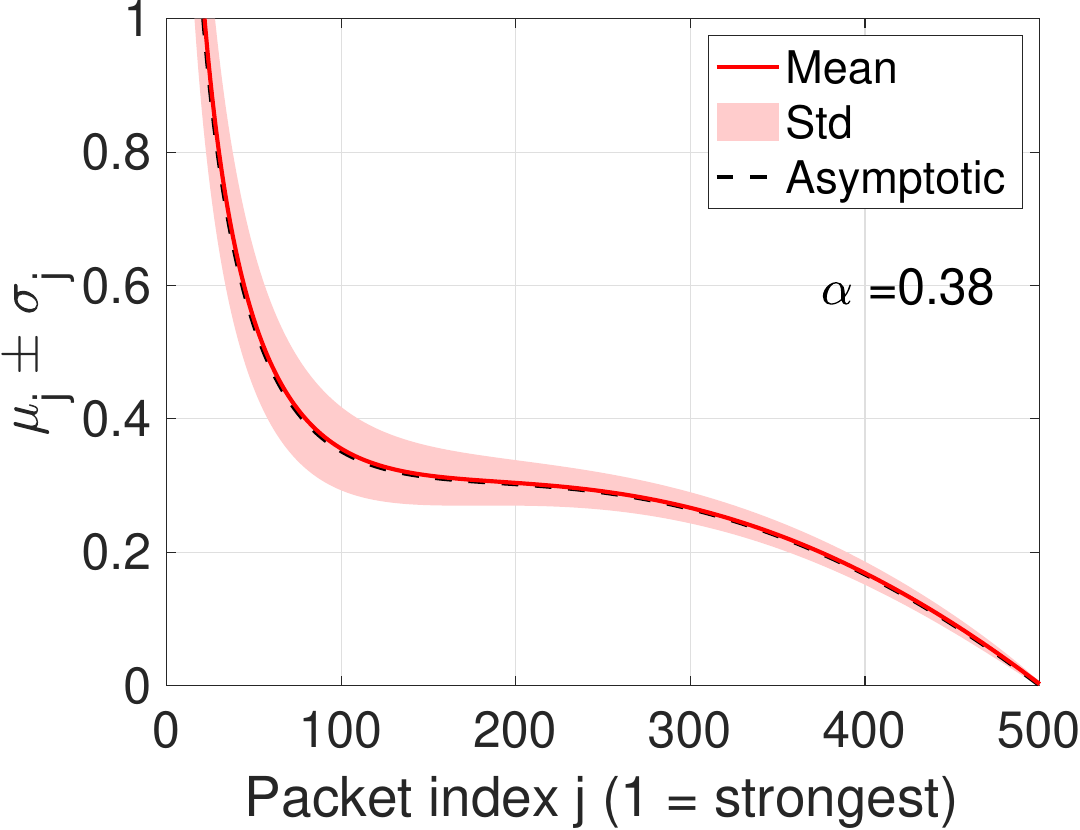} } \;%
	\subfloat[$n = 1000$]{ \label{fig:mupmsigman1000lin2} 
	    \includegraphics[width=.45\columnwidth]{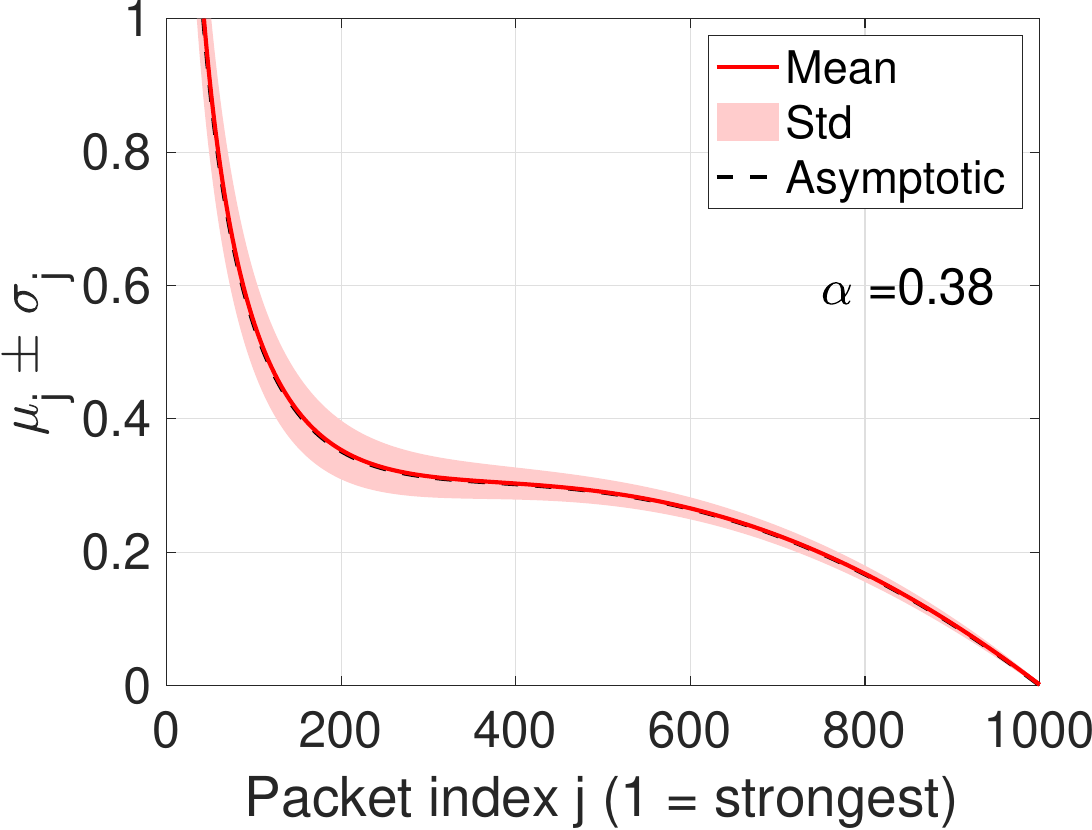} }
	\caption{Sequence of mean values $\mu_j(n)$ as a function of $x = j/n$ for several values of $n$ and $\xi = 0$. The shaded area corresponds to one standard deviation, i.e., $\mu_j(n) \pm \sigma_j(n)$. The dashed line is the function $f_{\alpha,\xi}(x)$ defined in \cref{eq:definitionfofx}. The required \ac{SNIR} threshold is set to $\gamma = 1/( \alpha n )$ with $\alpha = 0.38$}
	\label{fig:convergenceofVjlinalfa038}
\end{figure}

The same data is diaplyed in log scale in  \cref{fig:convergenceofVjlogalfa032,fig:convergenceofVjlogalfa038}.
This scale provides a better appreciation of the non-uniform convergence.
Convergence is more critical for smaller values of $\alpha$, i.e., for $\alpha < 1/e$, when local minimum and maximum exist.

\begin{figure}
	\centering
	\subfloat[$n = 50$]{ \label{fig:mupmsigman50log1}
	    \includegraphics[width=.45\columnwidth]{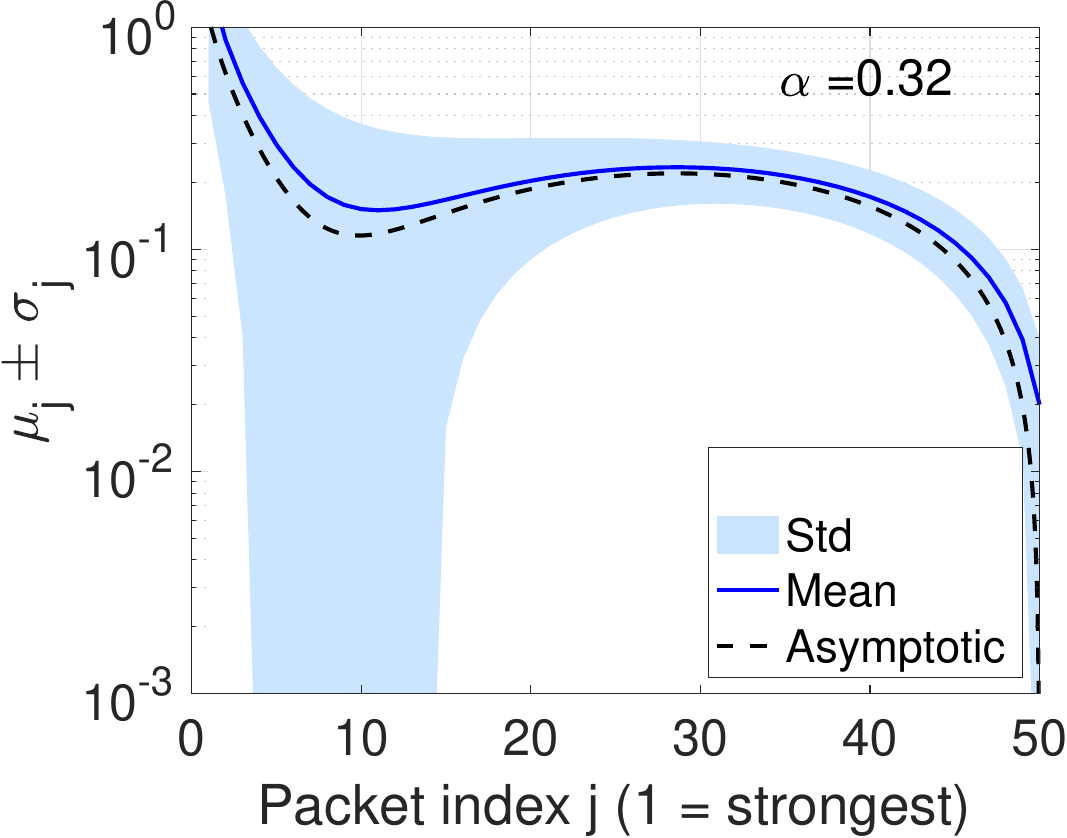} } \;%
	\subfloat[$n = 100$]{ \label{fig:mupmsigman100log1} 
	    \includegraphics[width=.45\columnwidth]{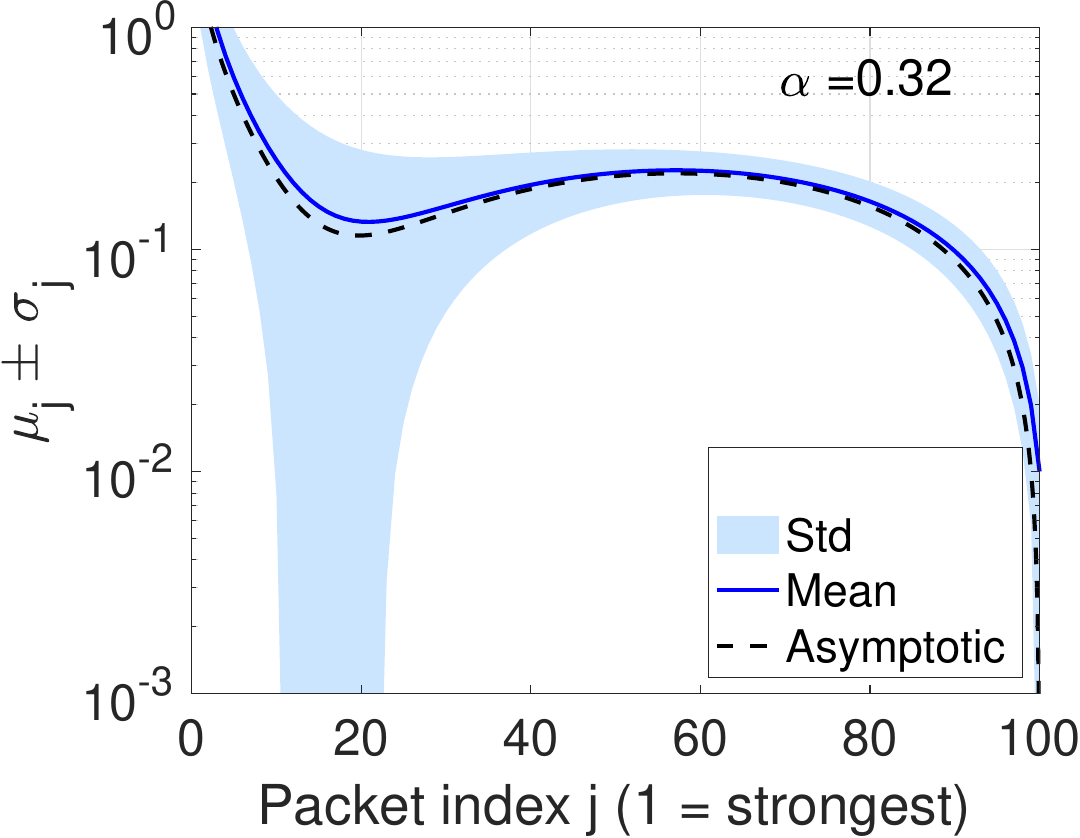} } \\
	\subfloat[$n = 500$]{ \label{fig:mupmsigman500log1}
	    \includegraphics[width=.45\columnwidth]{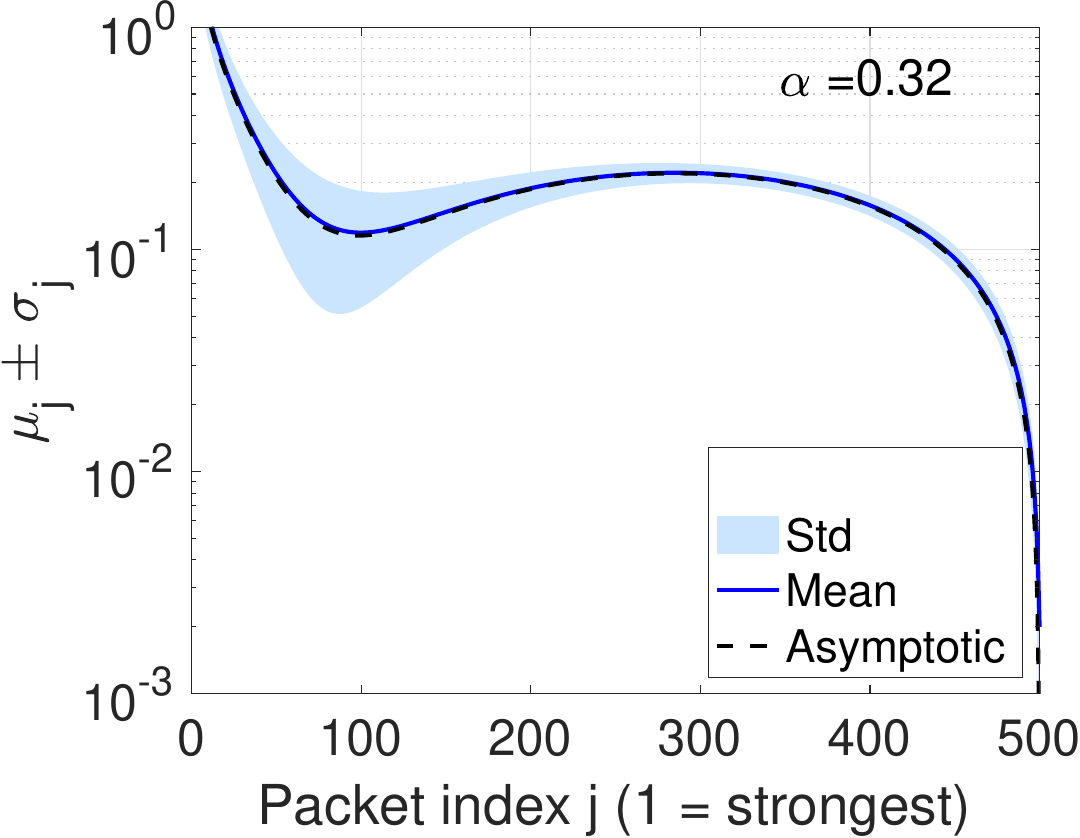} } \;%
	\subfloat[$n = 1000$]{ \label{fig:mupmsigman1000log1} 
	    \includegraphics[width=.45\columnwidth]{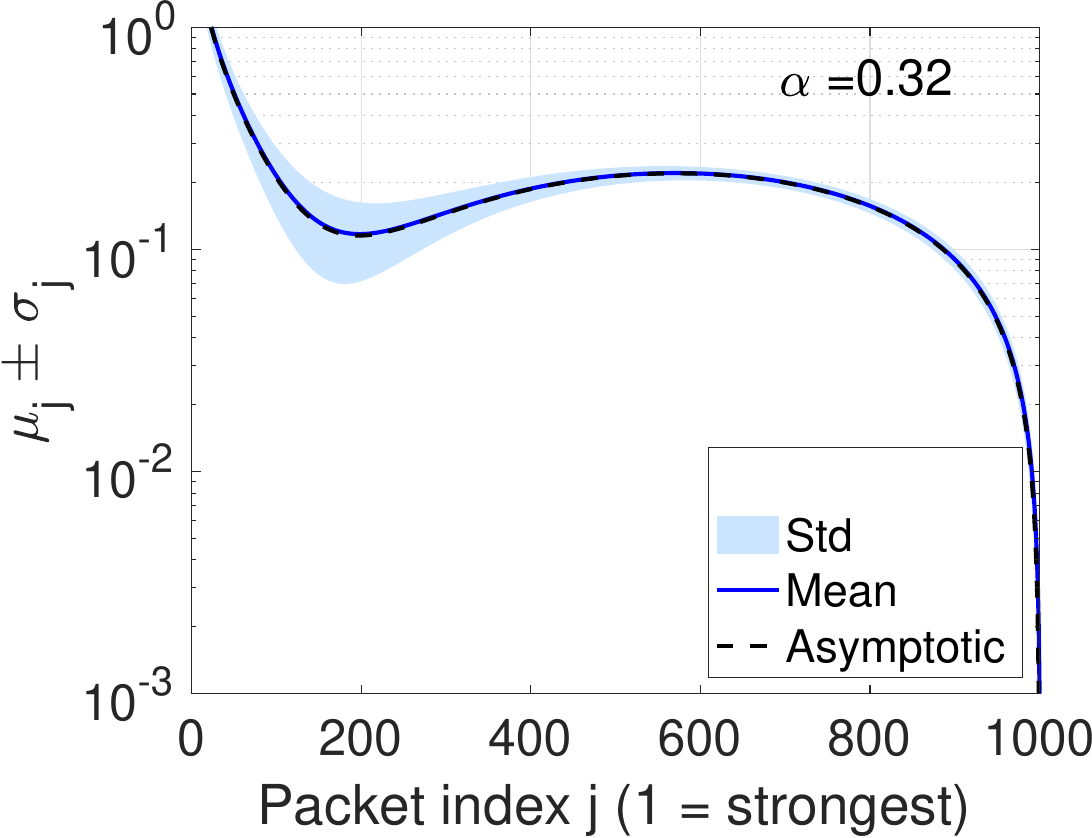} }
	\caption{Sequence of mean values $\mu_j(n)$ as a function of $x = j/n$ for several values of $n$ and $\xi = 0$. The shaded area corresponds to one standard deviation, i.e., $\mu_j(n) \pm \sigma_j(n)$. The dashed line is the function $f_{\alpha,\xi}(x)$ defined in \cref{eq:definitionfofx}. The required \ac{SNIR} threshold is set to $\gamma = 1/( \alpha n )$ with $\alpha = 0.32$}
	\label{fig:convergenceofVjlogalfa032}
\end{figure}

\begin{figure}
	\centering
	\subfloat[$n = 50$]{ \label{fig:mupmsigman50log2}
	    \includegraphics[width=.45\columnwidth]{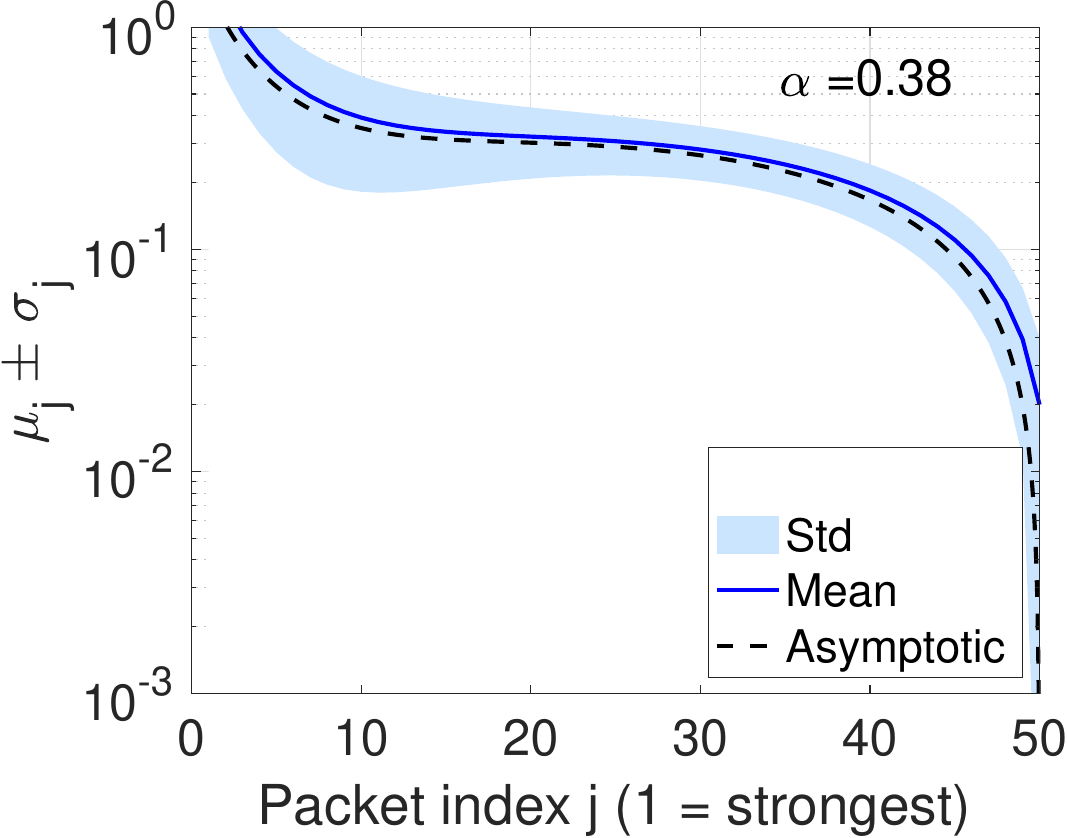} } \;%
	\subfloat[$n = 100$]{ \label{fig:mupmsigman100log2} 
	    \includegraphics[width=.45\columnwidth]{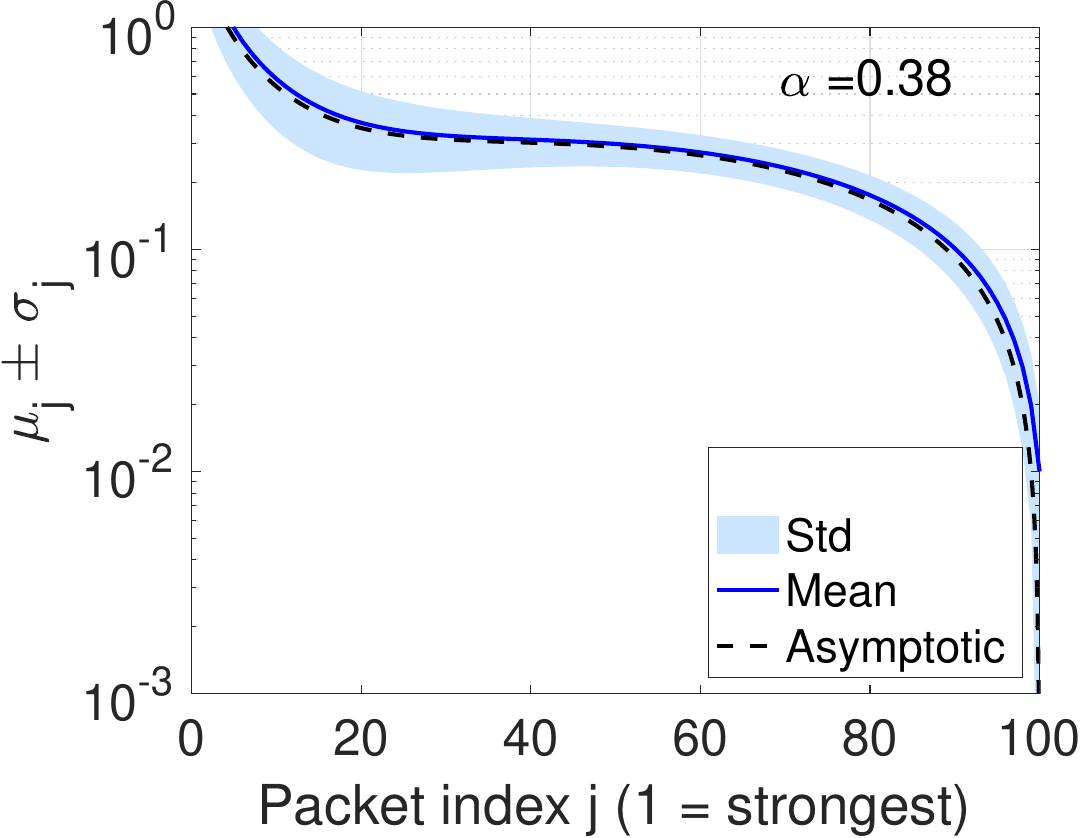} } \\
	\subfloat[$n = 500$]{ \label{fig:mupmsigman500log2}
	    \includegraphics[width=.45\columnwidth]{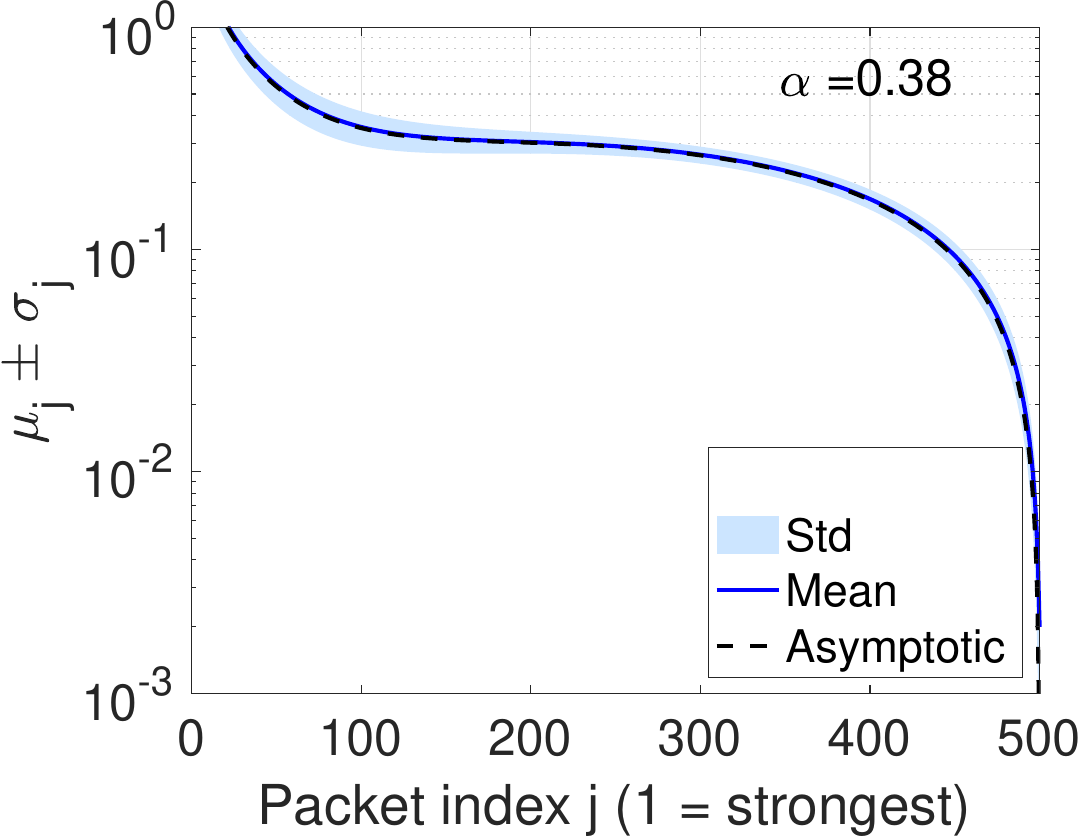} } \;%
	\subfloat[$n = 1000$]{ \label{fig:mupmsigman1000log2} 
	    \includegraphics[width=.45\columnwidth]{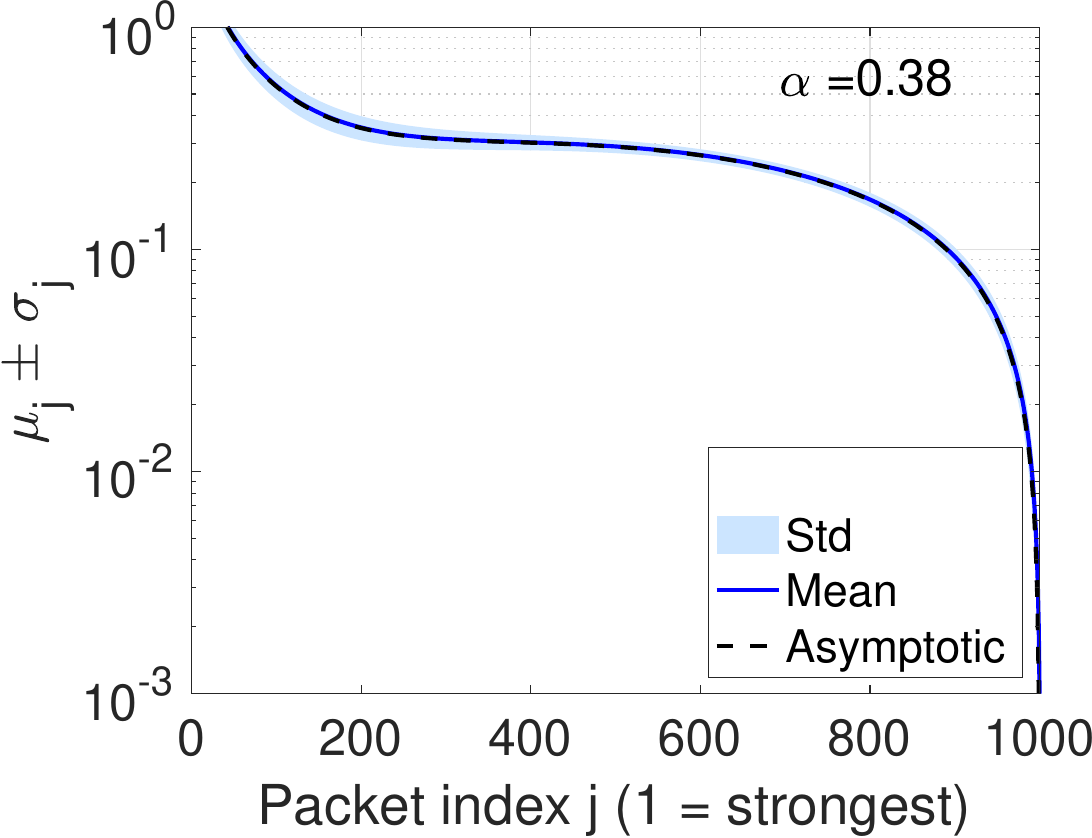} }
	\caption{Sequence of mean values $\mu_j(n)$ as a function of $x = j/n$ for several values of $n$ and $\xi = 0$. The shaded area corresponds to one standard deviation, i.e., $\mu_j(n) \pm \sigma_j(n)$. The dashed line is the function $f_{\alpha,\xi}(x)$ defined in \cref{eq:definitionfofx}. The required \ac{SNIR} threshold is set to $\gamma = 1/( \alpha n )$ with $\alpha = 0.38$}
	\label{fig:convergenceofVjlogalfa038}
\end{figure}

The main result stated below concerns the asymptotic behavior of the mean number of correctly decoded packets in a slot as $n \rightarrow \infty$.
We recall that $m_n(\gamma)$ denotes the mean number of correctly decoded packets in a slot with target \ac{SNIR} equal to $\gamma$, given that $n$ nodes transmit in that slot.
We have the following.

\begin{theorem}
\label{theo:mainresult}
Let $\gamma_n = 1/(\alpha n)$ for a positive constant $\alpha$.
Let also $f_{\alpha,\xi}(x) $ denote the function in \cref{eq:definitionfofx}.
It is $m_n(\gamma_n) \sim \zeta n$ where $\zeta$ is defined as\footnote{Note that the condition $f_{\alpha,\xi}(z) \ge c$ is equivalent to $g_\beta(z) \ge c + 1/\alpha$, where $g(\cdot)$ is defined in \cref{eq:definitionofgbetaofx}.}
\begin{equation}
\label{eq:definitionofzeta}
\zeta = \sup \left\{ x \in (0,1] : f_{\alpha,\xi}(z) \ge c, \forall z \in (0,x) \right\}
\end{equation}
for given values of $\alpha > 0$ and $\xi \in [0,1)$. 
\end{theorem}

\begin{proof}
See Appendix C.
\end{proof}

The meaning of this result is that it is possible to successfully decode a substantial (i.e., proportional to $n$) amount of packets among those transmitted simultaneously by the $n$ nodes, provided that  \ac{SIC} is used and the target \ac{SNIR} is set to a sufficiently small value, namely a value inversely proportional to the number of concurrent transmissions.

The achieved sum-rate is defined as the spectral efficiency of the channel.
The sum-rate gives the achieved spectral efficiency of the multiple access channel in bit/s/Hz \cite{Razzaque2022}.
Given that $n$ nodes transmit in one slot, the sum-rate is evaluated as follows:
\begin{align}
U_n(\gamma) &= \frac{ \mathrm{E}[ \text{ \# delivered bits per unit time } | \, n, \gamma ] }{ \text{ Channel bandwidth } }  \nonumber \\
  &= \frac{ L \, \mathrm{E}[ \text{ \# delivered packets per slot } | \, n, \gamma ] / T( \gamma ) }{ W }  \nonumber \\
  &=  \log_2(1+\gamma) \, \mathrm{E}[ \text{ \# delivered packets per slot } | \, n, \gamma ]  \nonumber \\
  &=  \log_2(1+\gamma) m_n(\gamma) 	\label{eq:sumratealltx}
\end{align}

Setting $\gamma = \gamma_n = 1/(\alpha n)$ in \cref{eq:sumratealltx}, we have in the limit for large $n$:
\begin{equation}
U_n(\gamma_n) = \log_2\left( 1 + \frac{ 1 }{ \alpha n } \right) m_n(\gamma_n) \sim \frac{ 1 }{ \alpha n \log 2 } \cdot \zeta n
\end{equation}
hence
\begin{equation}
\label{eq:Uasy}
\lim_{ n \rightarrow \infty }{ U_n(\gamma_n) } = \frac{ \zeta }{ \alpha \log 2 } = U_\infty
\end{equation}

This result gives a pathway to choosing a value for $\alpha$.
Specifically, $\alpha$ can be set to that value that maximizes the asymptotic value of the sum-rate in the limit for large $n$.
More in depth, according to the statement of \cref{theo:mainresult}, $\zeta$ is a function of $\alpha$, i.e., $\zeta = \zeta(\alpha)$, where $\zeta$ is defined in \cref{eq:definitionofzeta}. 
Then, in view of \cref{eq:Uasy}, we set $\alpha$ to the value that maximizes the ratio $\zeta(\alpha)/\alpha$.

\cref{fig:Unftyvsalfa} shows the asymptotic value of the sum-rate $U_\infty$ as a function of $\alpha$ for two values of $\xi$, namely, $\xi = 0$ on the left and $\xi = 0.1$ on the right.
In both cases, an optimal configuration of the parameter $\alpha$ is identified, that maximizes the asymptotic value of the sum-rate.
However, it is apparent that even a relatively small residual interference (10\% for $\xi = 0.1$) makes the achievable asymptotic sum-rate drop from almost 4 bit/s/Hz to about 2 bit/s/Hz.

\begin{figure}
	\centering
	\subfloat[$\xi = 0$]{ \label{fig:Unftyvsalfaxi0}
	    \includegraphics[width=.45\columnwidth]{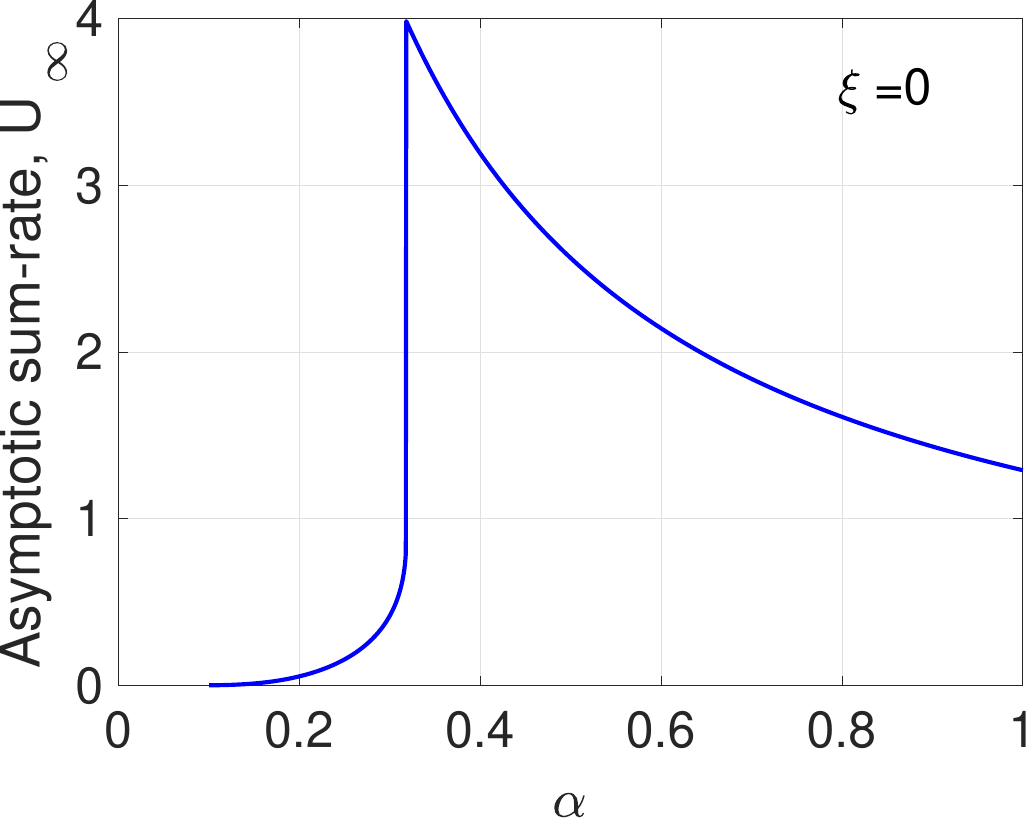} } \;%
	\subfloat[$\xi = 0.1$]{ \label{fig:Unftyvsalfaxi01} 
	    \includegraphics[width=.45\columnwidth]{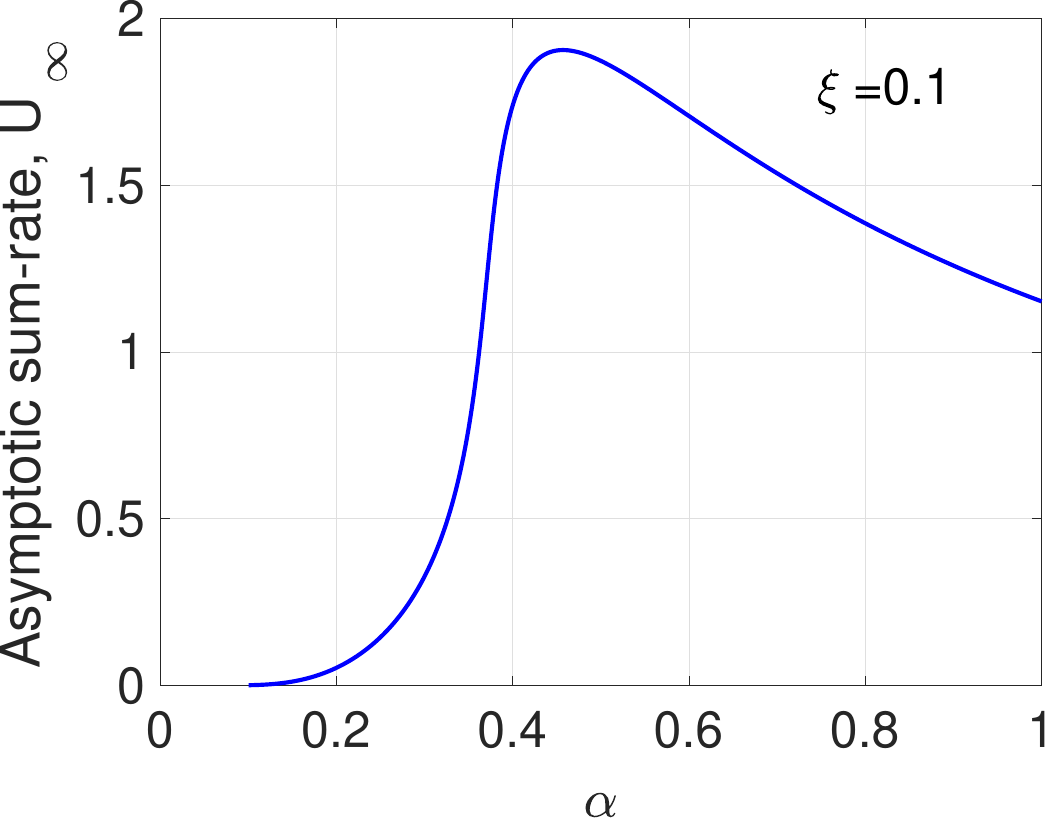} }
	\caption{Asymptotic value of the maximum sum-rate $U_\infty$ as a function of $\alpha$ for two values of $\xi$. The required \ac{SNIR} threshold is set to $\gamma = 1/( \alpha n )$.}
	\label{fig:Unftyvsalfa}
\end{figure}

The behavior of the optimal value of $\alpha$, denoted with $\alpha^*$, and the corresponding maximum value of the asymptotic sum-rate, denoted with $U_\infty^*$, are shown as a function of $\xi$ in \cref{fig:maxUasyandoptalfa}.

\begin{figure}
	\centering
	\subfloat[$\xi = 0$]{ \label{fig:optalfa}
	    \includegraphics[width=.45\columnwidth]{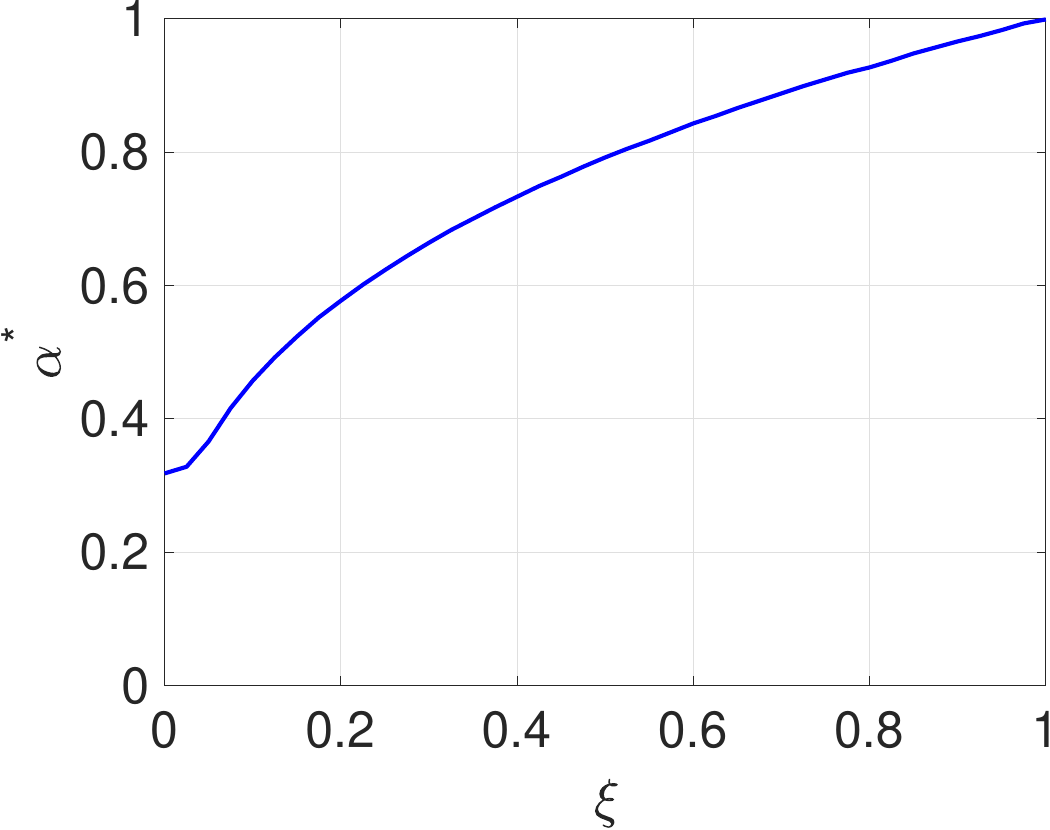} } \;%
	\subfloat[$\xi = 0.1$]{ \label{fig:maxUasy} 
	    \includegraphics[width=.45\columnwidth]{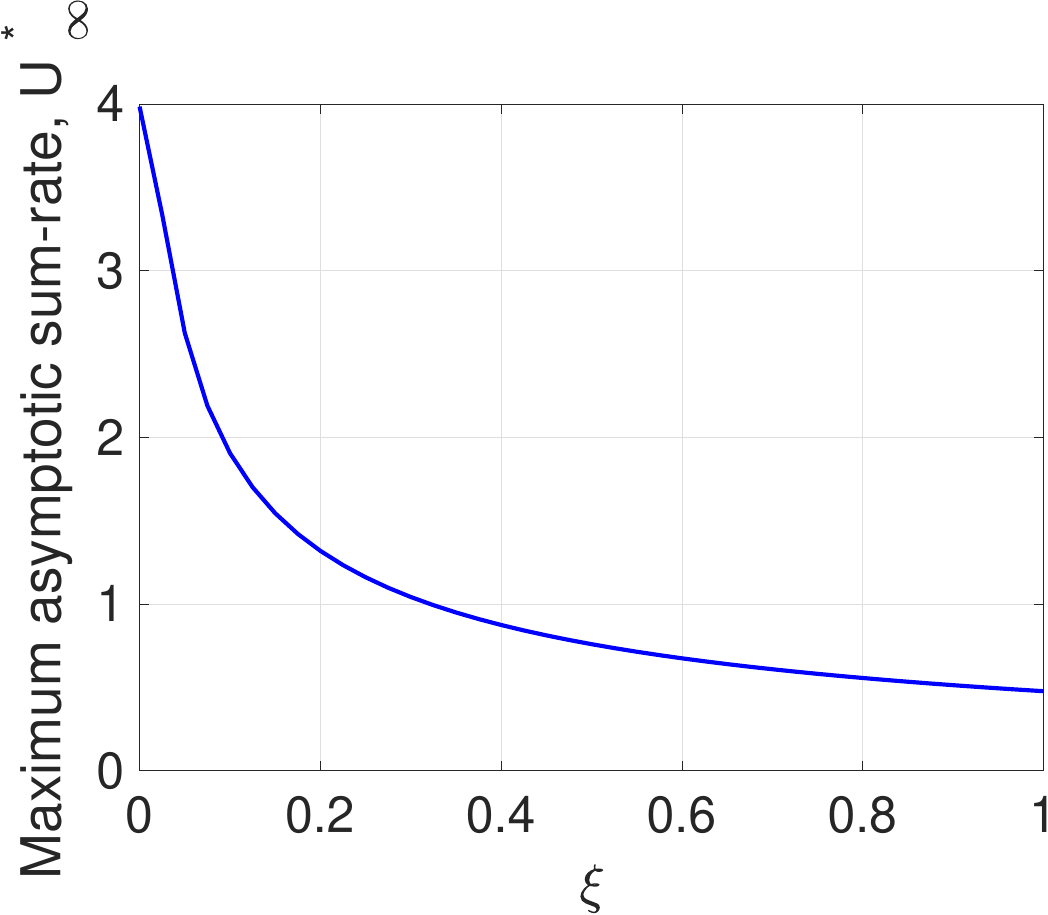} }
	\caption{Optimal configuration for maximum asymptotic sum-rate as a function of the coefficient $\xi$. (a) Optimal value of $\alpha$. (b) Maximum asymptotic sum-rate.}
	\label{fig:maxUasyandoptalfa}
\end{figure}

Finally, \cref{theo:caputrecase} gives an explicit form of the mean number of successfully decoded packets in case there is no \ac{SIC} (see \cref{eq:successfulnoSIC}).
This result is not new (e.g., see [CITATION]).
We report it here for ease of the reader, along with the asymptotic regime.

\begin{theorem}
\label{theo:caputrecase}
Let us consider $n$ concurrent transmissions towards a single receiver that can only exploit capture (no \ac{SIC}) with a target \ac{SNIR} $\gamma$.
The mean number of successfully decoded packets $m_n(\gamma)$ is given by
\begin{equation}
\label{eq:mnofgammaonlycapture}
m_n(\gamma) = e^{ - \gamma P_N / P_0 } \frac{ n }{ ( 1 + \gamma )^{n-1} } =  \frac{ ( 1 - \epsilon ) n }{ ( 1 + \gamma )^{n-1} }
\end{equation}
in view of the setting $P_0/P_N = -\gamma/\log( 1 - \epsilon )$.
With $\gamma = 1/( \alpha n )$, it is
\begin{equation}
U_\infty = \lim_{ n \rightarrow \infty }{ U_n } 
= \frac{ (1-\epsilon) }{ \log 2 } \frac{ 1 }{ \alpha } e^{ -1/\alpha }
\end{equation}
\end{theorem}

Setting $\gamma = 1/(\alpha n )$, it is easy to derive the asymptotic expansion $m_n(1/(\alpha n)) \sim n (1-\epsilon) \exp(-1/\alpha)$.
The asymptotic sum-rate in this case is maximized for $\alpha = 1$.
The maximum asymptotic sum-rate is $U_\infty^* = ( 1 - \epsilon )/( e \log 2)$.

The model with no \ac{SIC} applies to those multiple access schemes that mitigate interference, yet they do not cancel it actively, e.g., Code Division Multiple Access [REF] or Golden Modulation based on Zadoff-Chu sequences [REF].
As a matter of example, in the former case $\gamma = \gamma^\prime/PG$, where $\gamma^\prime$ is the target narrowband \ac{SNR} of a node (it is typically $\gamma^\prime \gg 1$) and $PG$ is the processing gain (typically $PG \gg 1$).

%

\section{Generalization}
\label{sec:discussion}

Let us now generalize the model of the random variable $S_i$, that represents the \ac{SNR} of node $i$.
Let us write $S_i = S_0 Y_i$, where $S_0$ is the \ac{SNR} level required to guarantee successful reception of a single transmitter with probability $1-\epsilon$.
In the analysis developed in previous Sections it has been assumed that $Y_i$ be a negative exponential random variable with mean 1, representing the Rayleigh fading gain.
We now generalize the analysis, by allowing $Y_i$ to be a general non-negative random variable, with continuous \ac{PDF}, still constrained to have mean equal to 1.
Sources of randomness in the description of $Y_i$ come from node path loss statistics and from the selected transmission power level.

Given $n$ nodes transmitting, we assume $Y_1,\dots,Y_n$ are \ac{i.i.d.} random variables with \ac{CDF} $F(t) = \mathcal{P}( Y \le t )$, \ac{CCDF} $G(t) = 1 - F(t)$, and \ac{PDF} $f(t)$.

The target average received \ac{SNR} $S_0$ is set so that $\mathcal{P}( S_0 Y \ge \gamma ) \ge 1 - \epsilon$.
Since $\mathcal{P}( Y > t ) = G(t)$, we have $G( \gamma/S_0 ) = 1-\epsilon$, i.e., $S_0 = \gamma/G^{-1}(1-\epsilon)$.
To keep notation uniform with the one used in previous section, we let $c = G^{-1}(1-\epsilon)$, so that we still wrtie $S_0 = \gamma/c$.
In case of Rayleigh fading it is $G(t) = e^{-t}$.
hence, $G^{-1}(y) = -\log y$, and we recover $c = -\log(1-\epsilon)$ (see \cref{eq:targetaveragerxSNR}).

The order statistics associated with $Y_1,\dots,Y_n$ is denoted with $Y_{(1)},\dots,Y_{(n)}$, where labeling is assigned so that $Y_{(1)} \ge \dots \ge Y_{(n)}$, i.e., in descending order.
It is known that \cite{Gnedenko1998}
\begin{equation}
f_{Y_{(h)}}(u) = \frac{ n! }{ (n-h)! (h-1)! } \left[ G(u) \right]^{h-1} \left[ 1 - G(u) \right]^{n-h} f(u)
\end{equation}
for $u \ge 0$ and $h = 1,\dots,n$.
The mean value of $Y_{(h)}$ is
\begin{align}
\mu_h &= \mathrm{E}[ Y_{(h)} ] = \int_{ 0 }^{ \infty }{ u f_{Y_{(h)}}(u) \, du }   \nonumber  \\
    &= n \int_{ 0 }^{ \infty }{ u \binom{ n-1 }{ h-1 } \left[ G(u) \right]^{h-1} \left[ 1 - G(u) \right]^{n-h} f(u) \, du }   \nonumber  \\
    &= n \int_{ 0 }^{ 1 }{ G^{-1}(u) \binom{ n-1 }{ h-1 } u^{h-1} (1 - u )^{n-h} \, du }  \\
    &= \int_{ 0 }^{ 1 }{ G^{-1}(u) \beta_{h,n-h+1}(u) \, du }
\end{align}
where $\beta_{a,b}(u)$ is the Beta \ac{PDF} with parameters $a$ and $b$, for $u \in [0,1]$ and $a,b > 0$.

Packet from node $h$ is decoded correctly if all packets sent by nodes whose signal was stronger (nodes from 1 to $h-1$) have been decoded correctly and if the \ac{SNIR} of node $h$ exceeds the threshold $\gamma$, i.e.,
\begin{equation}
\frac{ S_0 Y_{(h)} }{ 1 + \sum_{ i = h+1 }^{ n }{ S_0 Y_{(i)} } } \ge \gamma
\end{equation}

This inequality can be re-arranged as follows:
\begin{equation}
Y_{(h)} - \gamma \sum_{ i = h+1 }^{ n }{ Y_{(i)} } \ge c 
\end{equation}
where we have used the equality $S_0 = \gamma/c$.

For large $n$, we expect that the random variable on the left-hand side collapses to a deterministic distribution, hence it is equal to its mean with probability 1.
Let us therefore turn to the study of the inquality obtained by replacing the random variable on the left-hand side with its \emph{mean}.
\begin{equation}
\label{eq:inequalityofmeansgeneralcase}
\mu_h - \gamma \sum_{ i = h+1 }^{ n }{ \mu_i  } \ge c 
\end{equation}

We consider the asymptotic regime where $h \sim n x$, for a fixed $x \in (0,1)$, and $\gamma = 1/(\alpha n)$, with $\alpha > 0$.
Then, we have
\begin{align}
\mu_n &= \int_{ 0 }^{ 1 }{ G^{-1}(u) \beta_{n x,n(1-x)+1}(u) \, du }   \nonumber  \\
    &\sim  \int_{ 0 }^{ 1 }{ G^{-1}(u) \beta_{n x,n(1-x)}(u) \, du }
\end{align}
where we account for the fact that $n(1-x)+1 \sim n(1-x)$ for large $n$ and fixed $x$.

We use the following asymptotic property of the beta probability distribution:
\begin{equation}
\beta_{n a, n b}(u) \sim \mathcal{N}\left( \frac{ a }{ a+b } , \frac{ 1 }{ n } \, \frac{ a b }{ ( a + b )^2 } \right)
\end{equation}
Therefore, in our case we get
\begin{equation}
\beta_{n x, n (1-x)}(u) \sim \mathcal{N}\left( x , \frac{ x (1-x) }{ n } \right)
\end{equation}
This shows that as $n \rightarrow \infty$, the scaled Beta distribution tends to a Dirca impulse centered in $x$.
As a consequence, we find as $n$ tends to infinity:
\begin{equation}
\label{eq:inequalityonmeanvalues}
\mu_n \rightarrow \int_{ 0 }^{ 1 }{ G^{-1}(u) \delta( u - x ) \, du } = G^{-1}(x)
\end{equation}

As for the sum appearing on the right-hand side of \cref{eq:inequalityofmeansgeneralcase}, we have, for $\gamma = 1/( \alpha n )$,
\begin{equation}
\label{eq:expressionofthesum}
\gamma \sum_{ i = h+1 }^{ n }{ \mu_i  } = \frac{ 1 }{ \alpha } \int_{ 0 }^{ 1 }{  G^{-1}(u) \sum_{ i = h }^{ n-1 }{ B_{u,n-1}(i) } \, du }
\end{equation}
where
\begin{equation}
\label{eq:binompdfB}
B_{u,n-1}(i) = \binom{ n-1 }{ i } u^i (1-u)^{n-1-i} \, , \quad i = 0,1,\dots,n-1,
\end{equation}
is a Binomial probability distribution with parameters $u$ and $n-1$.
Let $Z$ denote the binomial random variable with probability distribution in \cref{eq:binompdfB}.
Then, we have
\begin{equation}
\sum_{ i = h }^{ n-1 }{ B_{u,n-1}(i) } = \mathcal{P}( Z \ge h )
\end{equation}

In the considered asymptotic regime, it is $h \sim n x$.
Moreover, we use the following asymptotic property of the Binomial probability distribution:
\begin{equation}
\label{eq:binomialasy}
\lim_{ n \rightarrow \infty }{ \mathcal{P}\left( \frac{ Z - n u }{ \sqrt{ n u ( 1-u ) } } \le t \right) }= \Phi(t)
\end{equation}
where $\Phi(t)$ is the \ac{CDF} of the standard Gaussian random variable:
\begin{equation}
\Phi(t)  = \frac{ 1 }{ \sqrt{ 2 \pi } } \int_{ -\infty }^{ t }{ e^{ - u^2 / 2 } \, du }
\end{equation}
This stems from the Central Limit Thoerem and from the fact that
\begin{equation}
Z = \sum_{ k = 1 }^{ n-1 }{ B_k }
\end{equation}
where $B_k$ is a Bernoulli random variable that equals 1 with probability $u$ and it is equal to 0 with probability $1-u$.

From \cref{eq:binomialasy}, we derive that
\begin{equation}
\mathcal{P}( Z \le n x ) \sim \Phi\left( \frac{ n x - n u }{ \sqrt{ n u ( 1-u ) } } \right) = \Phi\left( \sqrt{n} \frac{ x - u }{ \sqrt{ u ( 1-u ) } } \right) 
\end{equation}

As $n$ gets large, it is seen that $\Phi\left( \sqrt{n} \frac{ x - u }{ \sqrt{ u ( 1-u ) } } \right)  \rightarrow 0$, if $x < u$, whereas $\Phi\left( \sqrt{n} \frac{ x - u }{ \sqrt{ u ( 1-u ) } } \right)  \rightarrow 1$, if $x > u$.
Therefore, the probability of the event $Y \le h = n x$ tends to a step function with jump located at $x$ as $n$ grows.

Summing up, we have
\begin{equation}
\sum_{ i = n x }^{ n-1 }{ B_{u,n-1}(i) } = \mathcal{P}( Y \ge n x ) \rightarrow \begin{cases}
    1 & u > x \text{ }, \\
    0  & u < x.
\end{cases}
\end{equation}
Inserting this asymptotic result in the integral in \cref{eq:expressionofthesum}, we have as $n$ tends to infinity:
\begin{equation}
\gamma \sum_{ i = n x+1 }^{ n }{ \mu_i  } \rightarrow \frac{ 1 }{ \alpha } \int_{ x }^{ 1 }{  G^{-1}(u) \, du }
\end{equation}
In the end, the inequality on the mean values in \cref{eq:inequalityonmeanvalues} approaches the following limiting inequality in the asymptotic regime
\begin{equation}
\label{eq:inequalityonmeanvaluesasymptotic}
G^{-1}(x) - \frac{ 1 }{ \alpha } \int_{ x }^{ 1 }{  G^{-1}(u) \, du } \ge c
\end{equation}
for a given $x \in (0,1)$.
The quantity $x$ represents the fraction of correctly decoded packets.
The asymptotic sum rate is $U_\infty = \frac{ x^*(\alpha) }{ \alpha \log 2 }$, where $x^*(\alpha)$ is the smallest value of $x \in (0,1)$ such that the inequality in \cref{eq:inequalityonmeanvaluesasymptotic}  is satisfied.

In the special case where $Y$ is a negative exponential random variable with mean 1, it is $G(t) = e^{-t}$, hence $G^{-1}(u) = - \log u$.
It can be verified that \cref{eq:inequalityonmeanvaluesasymptotic} reduces to \cref{eq:definitionfofx}.

\cref{eq:inequalityonmeanvaluesasymptotic} can be transformed as follows:
\begin{equation}
\label{eq:finalinequalitygeneralized}
y - \frac{ 1 }{ \alpha } \int_{ 0 }^{ y }{  u f(u) \, du } \ge c = G^{-1}( 1 - \epsilon )
\end{equation}
where $y =  G^{-1}(x)$, hence it is $x = G(y)$.
The fraction of transmission that can be correctly decoded is $x^* = G(y^*)$, with $y^*$ equal to the largest value of $y$ such that the left-hand side of \cref{eq:finalinequalitygeneralized} is less than $c$.
Formally,
\begin{equation}
y^* = \inf \left\{ y : z \ge c + \int_{0}^{z}{ u f(u) du }, \forall z \in (y,\infty) \right\}
\end{equation}

Since $f(u) \ge 0$ and the mean of the \ac{PDF} $f(u)$ is constrained to be 1, it is easy to verify that it must be $c \le y^* \le c+1/\alpha$.
Since $x^* = G(y^*)$ and $G(\cdot)$ is monotonously decreasing, we have $G\left( G^{-1}( 1 - \epsilon ) + 1/\alpha \right) \le x^* \le G\left( G^{-1}( 1 - \epsilon ) \right) = 1-\epsilon$.

To provide a numerical example, we assume $Y$ has a Gamma \ac{PDF}, one of the simplest probability distributions that provides two degrees of freedom.
We set the mean of $Y$ to 1 and let $\eta$ denote the reciprocal of the \ac{SCOV} of $Y$, i.e., $\eta$ is the ratio of the mean of $Y$ squared to the variance of $Y$.
As $\eta$ gets smaller, the random variable $Y$ exhibits more variability, which is deemed to be favorable to the success of \ac{SIC} decoding.

Formally, for any $\eta > 0$, we have
\begin{equation}
f(t) = \frac{ \eta^\eta t^{\eta-1} }{ \Gamma(\eta) } e^{ - \eta t } \, , \quad t > 0.
\end{equation}
where $\Gamma(a)$ is the Euler Gamma function, defined by $\Gamma(a) = \int_{ 0 }^{ \infty }{  u^{a-1} e^{-u} \, du }$ for $a > 0$.
With this \ac{PDF}, \cref{eq:finalinequalitygeneralized} specializes into
\begin{equation}
\label{eq:yequatiowithgammapdf}
y - \frac{ 1 }{ \alpha } \frac{ \Gamma( \eta+1 , \eta y ) }{ \Gamma( \eta+1) } \ge c = \frac{ \Gamma^{-1}\left( \eta, \epsilon \Gamma(\eta) \right) }{ \eta }
\end{equation}
where $\Gamma(a,x)$ is the incomplete Euler Gamma Function, given by $\Gamma(a,x) = \int_{ 0 }^{ x }{ u^{a-1} e^{-u} \, du }$ for $a,x > 0$, and $\Gamma^{-1}(a,\cdot)$ is its inverse with respect to the second variable, $x$.
Note $\Gamma(a) = \Gamma(a,\infty)$ and $\Gamma(a,x)$ is a monotonously increasing function of $x$ for any positive $a$.
We have also used the equality $\Gamma(\eta+1) = \eta \Gamma(\eta)$, holding for any $\eta > 0$.

Exploiting the equation $U_\infty = x^*/( \alpha \log(2) )$, it is possible to find numerical values of the asymptotic sum-rate $U_\infty$ as a function of $\alpha$, for several values of the \ac{SCOV} of $Y$ (see \cref{fig:Uinfvsalfageneralized}).

\begin{figure}
	\centering
	\subfloat{ \label{fig:Uinfvsalfageneralized}
	    \includegraphics[width=.45\columnwidth]{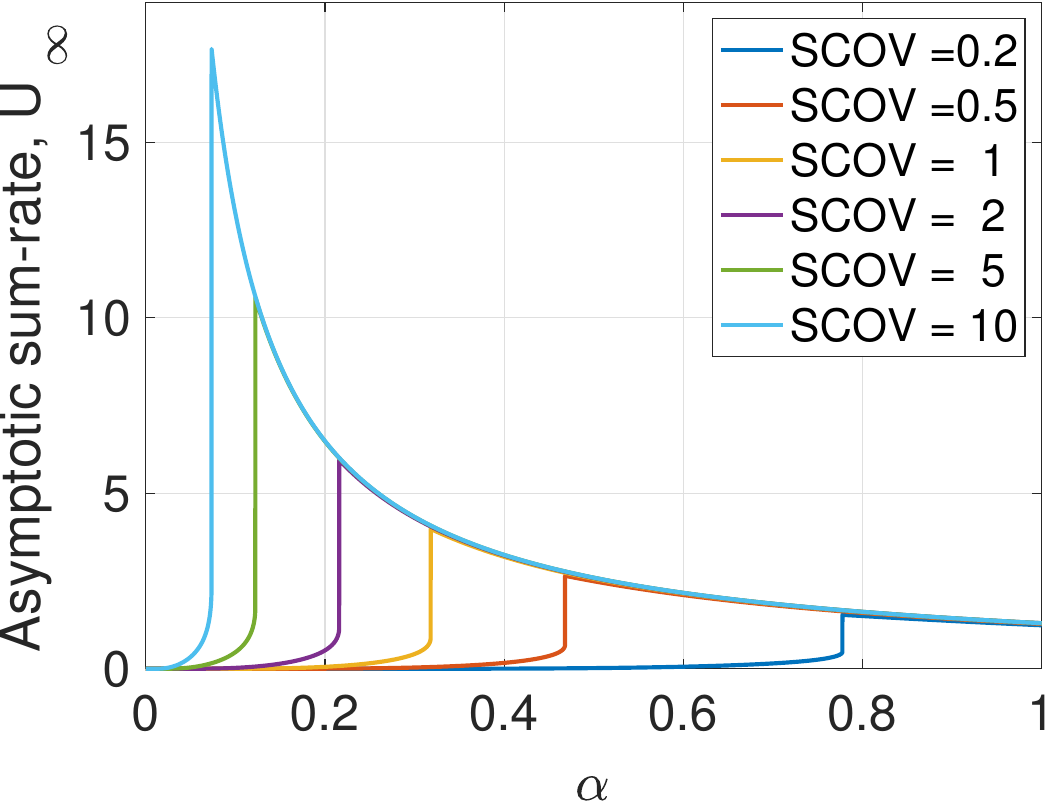} } \;%
	\subfloat{ \label{fig:UinfstarvsSCOVgeneralized} 
	    \includegraphics[width=.45\columnwidth]{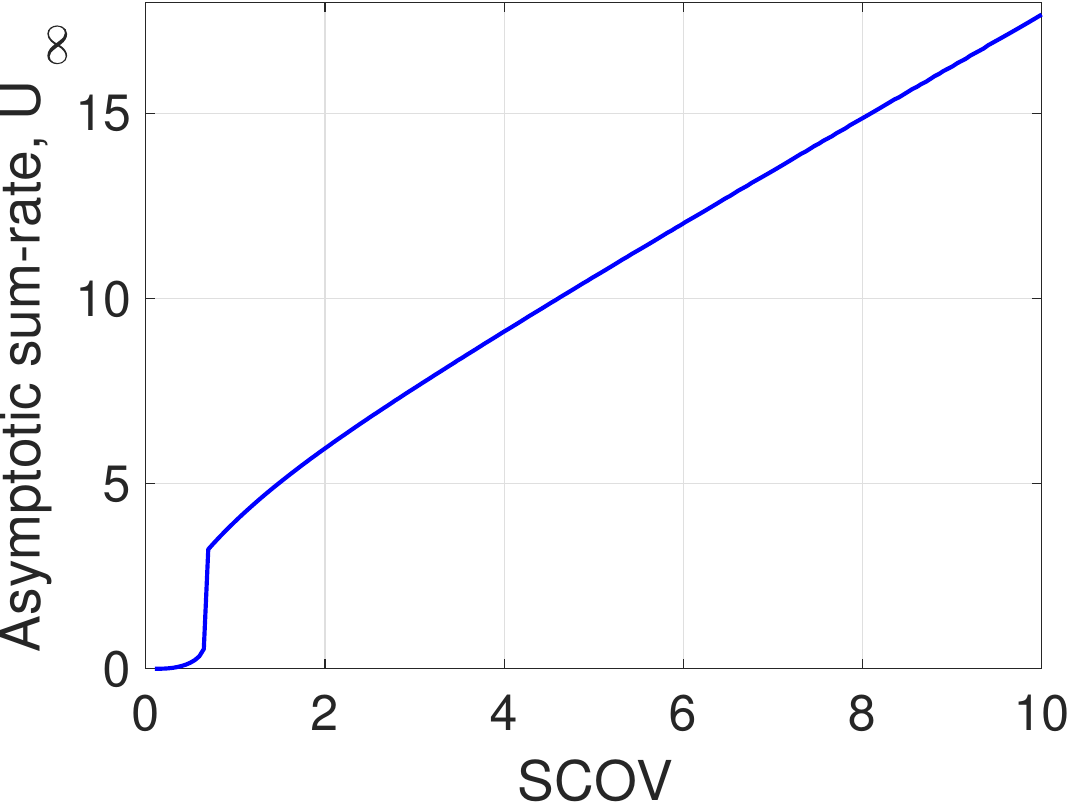} }
	\caption{Analysis of the asymptotic sum-rate in case of generalized fading with Gamma \ac{PDF}. 
	(a) Asymptotic sum-rate as a function of $\alpha$ for several values of the \ac{SCOV} of $Y$. (b) Maximum achievable asymptotic sum-rate as a function of the \ac{SCOV} of $Y$.}
	\label{fig:sumrategeneralized}
\end{figure}

It appears that, for any fixed value of the \ac{SCOV}, there is an optimal value of $\alpha$ that maximizes the asymptotic sum-rate.
The behavior of the maximum achievable asymptotic sum-rate as a function of the \ac{SCOV} is shown in \cref{fig:UinfstarvsSCOVgeneralized}. 
This results confirm the intuition that suggests a higher variability of $Y$ (hence of the received power level) is beneficial to the success of \ac{SIC} receiver.

In General, we can write $Y = X Z$, where $X$ represents the random component of the block fading channel and $Z$ stems from a randomization of the transmission power level.
We remind that it must be $\mathrm{E}[ Y ] = 1$.
If we assume Rayleigh fading, so that $\mathcal{P}( X > u ) = e^{-u}$, then it must be $\mathrm{E}[ Z ] = 1$.
With a two level choice for transmission power, we have $Z = 1/b$ with probability $b/(1+b)$ and $Z = b$ with probability $1/(1+b)$, for any $b \ge 1$.
It can be verified that $\mathrm{E}[ Y^2 ] = 2 \mathrm{E}[ Z^2 ] = 2 \left( b + 1/b - 1 \right)$.
Hence, the \ac{SCOV} of $Y$ is $2 \left( b + 1/b - 1 \right) - 1$, which can be made as large as desired with growing values of $b$.
To ease transmission power setting, the nodes could be instructed to opt for the largest transmission power level in case their average channel gain is above a suitable threshold, for the lower value otherwise.
Setting the threshold properly assigns probability of using either level as desired (i.e., so that the average of $Z$ be 1).

%

\section{Conclusions}
\label{sec:conclusions}

We have characterized the asymptotic behavior of the mean number of correctly decoded packets and the ensuing sum-rate of a \ac{SIC} receiver with $n$ transmitting nodes and a target \ac{SNIR} that scales in inverse proportion to $n$
The considered \ac{SIC} is quite general, encompassing any kind of random fading.
The average received power level is assumed to be same for all nodes.
It is found that the random variable representing the number of correctly decoded packets converges to a deterministic quantity as $n \rightarrow \infty$.
The asymptotic regime is fully characterized.

Two major points for further work stem out of the result of this paper.
First, the generalization to systems where nodes play with transmission power levels, so as to obtain different  average received power levels, are of interest.
In fact, \ac{SIC} performance are expected to improve, if superposed signals have well separated reception power levels.
Secondly, practical techniques for realizing \ac{SIC} should be evaluated.
Referring to specific \ac{SIC} algorithms allows the identification of a refined reception model.
It could also provide evidence of the feasibility of large scale interference cancellation receivers.

\appendices

\section{Proof of bounds on $m_k(\gamma)$}
\label{app:Appboundsonmkgamma}
The lower bound is obtained by assuming no \ac{SIC} is carried out.
Let $k$ nodes transmit in a slot and let $S_h = Y_h \gamma/c$ the \ac{SNR} of node $h$ at the receiving \ac{BS}, where $Y_h$ is the Rayleigh fading gain.

With no \ac{SIC}, the \ac{SNIR} of node $h$ at the receiver is
\begin{equation}
\frac{ Y_h \gamma / c }{ 1 + \sum_{ j \ne h }{ Y_j \gamma / c } } \le \gamma
\end{equation}
Node $h$ is decoded successfully if 
\begin{equation}
Y_h \ge c + \gamma \sum_{ j \ne h }{ Y_j } 
\end{equation}
where all $Y_i$'s are i.i.d. negative exponential random variables with mean 1.
Let us define the following events for $h = 1,\dots,k$:
\begin{equation}
E_h = \left\{ Y_h \ge c + \gamma \sum_{ j \ne h }^{ k }{ Y_j } \right\}
\end{equation}

Conditioning on $Y_j = y_j$ for $j \ne h$, we have
\begin{equation}
\mathcal{P}\left( Y_h \ge c + \gamma \sum_{ j \ne h }{ y_j } \right) = e^{ - c } \prod_{ j \ne h }{ e^{ - \gamma y_j } }
\end{equation}
Removing the conditioning and reminding that $c = - \log(1-\epsilon)$, we find
\begin{equation}
\mathcal{P}( E_h ) = \frac{ 1-\epsilon }{ (1+\gamma)^{k-1} }
\end{equation}
The number of successfully decoded packets is the sum of indicator functions of the events $Y_h \ge c + \gamma \sum_{ j \ne h }{ Y_j }$, for $h = 1,\dots,n$.
Hence, the mean number of successfully decoded packets is
\begin{align}
m_k(\gamma) &\ge m^{\text{no-SIC}}_k(\gamma) = \mathrm{E}\left[ \sum_{ h = 1 }^{ n }{ I( E_h ) } \right]   \nonumber \\
  &= \sum_{ h = 1 }^{ n }{ \mathcal{P}( E_h ) }  = \frac{ k (1-\epsilon) }{ (1+\gamma)^{k-1} }  \label{eq:lowerboundonmkgamma}
\end{align}
where $I(E)$ is the indicator function of event $E$, equal to 1 if and only if event $E$ occurs.

As for the upper bound, let us consider the order statistics associated with the i.i.d. random variables $Y_j$, denoted as $\tilde{Y}_j$, with $\tilde{Y}_1 \ge \tilde{Y}_2 \ge \dots \ge \tilde{Y}_k$. 

At least the $h$ strongest packets are successfully decoded with \ac{SIC}, if the following inequalities are simultaneously satisfied:
\begin{equation}
\tilde{Y}_r \ge c + \gamma \sum_{ j = r+1 }^{ k }{ \tilde{Y}_j } \, , \qquad r = 1,\dots,h. 
\end{equation}
Let us define the following events for $h = 1,\dots,k$:
\begin{equation}
\tilde{E}_h = \left\{ \tilde{Y}_h \ge c + \gamma \sum_{ j = h+1 }^{ k }{ \tilde{Y}_j } \right\} 
\end{equation}
Let also $\tilde{D}_k$ be the number of successfully decoded packets with \ac{SIC}, given that $k$ nodes transmit.
For each $h = 1,\dots,k$, we have
\begin{equation}
\mathcal{P}( \tilde{D}_k \ge h ) = \mathcal{P}\left( \bigcap_{ r = 1}^{ h }{ \tilde{E}_r } \right) \le \mathcal{P}( \tilde{E}_1 )
\end{equation}
From the very definition of the event $\tilde{E}_1$ and $E_h, \, h = 1,\dots,k$, it follows that $\tilde{E}_1 \subseteq \bigcup_{ h = 1 }^{ k }{ E_h }$.
Then, exploiting the union bound, we find:\footnote{Based on the representation of the order statistics $\tilde{Y}_1,\dots,\tilde{Y}_k$ given in Appendix A of \cite{ASYMPTOTICPAPERARXIV}, it is possible to show that $\mathcal{P}( \tilde{E}_1 ) = m^{\text{no-SIC}}_k(\gamma)$ in case of Rayleigh fading.}
\begin{align*}
\mathcal{P}( \tilde{D}_k \ge h ) &\le \mathcal{P}( \tilde{E}_1 ) \le \mathcal{P}\left( \bigcup_{ h = 1 }^{ k }{ E_h } \right)   \\
  &\le \sum_{ h = 1 }^{ k }{ \mathcal{P}( E_h ) } = m^{\text{no-SIC}}_k(\gamma)
\end{align*}
Finally, we have
\begin{align}
m_k(\gamma) &= \mathrm{E}[ \tilde{D}_k ] = \sum_{ h = 1}^{ k }{ \mathcal{P}( \tilde{D}_k \ge h ) }   \nonumber \\
  &\le k \, m^{\text{no-SIC}}_k(\gamma) = \frac{ k^2 (1-\epsilon) }{ (1+\gamma)^{k-1} }  \label{eq:upperboundonmkgamma}
\end{align}

This upper bound proves that
\begin{equation}
\limsup_{ k \rightarrow \infty }{ m_k(\gamma) } = 0
\end{equation}
Moreover, relaxing the integer $k$ to a real positive variable, it can be verified that the rightmost term in \cref{eq:upperboundonmkgamma} is maximized for $k = 2/\log(1_\gamma)$.
Hence:
\begin{equation}
\label{eq:maximumofmkofgamma}
m_k(\gamma) \le M(\gamma) = \frac{ (1-\epsilon) (1+\gamma) }{ e^2 \left[ \log(1+\gamma) \right]^2 } \, , \qquad \forall k.
\end{equation}

Exploiting those bounds, we can find bounds for the mean number $M_n(p,\gamma)$ of successfully decoded packets in a slot, given that $n$ nodes are backlogged at the beginning of the slot.
Since it is
\begin{equation}
M_n(p,\gamma) = \sum_{ k = 1 }^{ n }{ m_k(\gamma) \binom{n}{k} p^k (1-p)^{n-k} }
\end{equation}
using \cref{eq:lowerboundonmkgamma,eq:upperboundonmkgamma}, we have:
\begin{equation}
A_n( p , \gamma ) \le M_n( p, \gamma ) \le A_n( p , \gamma ) + B_n( p , \gamma )
\end{equation}
with
\begin{align*}
&A_n( p , \gamma ) = ( 1 - \epsilon ) n p \left( 1 - \frac{ \gamma }{ 1 + \gamma } p \right)^{n-1}   \\
&B_n( p, \gamma ) = ( 1 - \epsilon ) \frac{ n (n-1) p^2 }{ 1 + \gamma } \left( 1 - \frac{ \gamma }{ 1 + \gamma } p \right)^{n-2}
\end{align*}

The value of $p$ at which the lower bound is maximized is easily found to be $p^* = \frac{ 1 + \gamma }{ n \gamma }$, for $n \ge 1 + 1 /\gamma$.
With this value of $p$, the bounds on $M_n( p , \gamma )$ are as follows:
\begin{equation}
\label{eq:finitenboundsonMnpandgamma1}
1 \le \frac{ M_n\left( \frac{ 1 + \gamma }{ n \gamma } , \gamma \right) }{ A_n\left( \frac{ 1 + \gamma }{ n \gamma } , \gamma \right) } \le  1 + \frac{ 1 }{ \gamma } 
\end{equation}
with
\begin{equation}
\label{eq:finitenboundsonMnpandgamma2}
A_n\left( \frac{ 1 + \gamma }{ n \gamma } , \gamma \right)  = ( 1 - \epsilon ) \frac{ 1 + \gamma }{ \gamma } \left( 1 - \frac{ 1 }{ n } \right)^{n-1} 
\end{equation}

As $n \rightarrow \infty$, we get finally:\footnote{Under any scaling of the type $p_n = a/n$ for positive $a$, the limit of $M_n(p_n,\gamma)$ for fixed $\gamma$ exists, since the binomial distribution with parameters $n$ and $p_n = a/n$ tends to a Poisson distribution with mean $a$.}
\begin{equation}
\label{eq:asymptoticboundsonMnpandgamma}
 1 \le \frac{ \lim_{ n \rightarrow \infty }{ M_n\left( \frac{ 1 + \gamma }{ n \gamma } , \gamma \right) } }{ ( 1 - \epsilon ) \frac{ 1 + \gamma }{ \gamma \, e } } \le 1 + \frac{ 1 }{ \gamma } 
\end{equation}

These inequalities are especially useful when $\gamma \gg 1$.

\section{Proof of instability with fixed parameters}
\label{app:AppUfixedparameters}
In this Appendix we give a proof of the instability of $Q(t)$ in case $p_n$ and $\gamma_n$ are set to constant values $p$ and $\gamma$ respectively for all $n$.
We will prove that, under this setting, $U_n(p,\gamma) \rightarrow 0$ as $n \rightarrow \infty$.
Then, the conditional drift is positive for all $n$, but at most a finite set of values, for any positive $\Lambda$.

Given an arbitrary $\varepsilon > 0$, let $K_\varepsilon = n p  - \sqrt{ \frac{ n p (1-p) }{ \varepsilon } }$.
By Chebichev's inequality, applied to the binomial random variable $B_n$ with parameters $n$ and $p$, we have
\begin{align}
   \mathcal{P}( B \le K_\varepsilon ) &= \mathcal{P}\left( B - np \le - \sqrt{ \frac{ n p (1-p) }{ \varepsilon } } \right)   \\
    &\le \mathcal{P}\left( | B - np | \ge  \sqrt{ \frac{ n p (1-p) }{ \varepsilon } } \right)   \\
    &\le \frac{ n p (1-p) }{ \left( \sqrt{ \frac{ n p (1-p) }{ \varepsilon } } \right)^2 }  = \varepsilon
\end{align}

In Appendix \ref{app:Appboundsonmkgamma} it is shown that, for any fixed $\gamma$, $\limsup_{ k \rightarrow \infty }{ m_k(\gamma) } = 0$  and $\max_{ k \ge 0 }{ m_k(\gamma) } = M(\gamma) < \infty$ (see \cref{eq:maximumofmkofgamma}).
There exists an integer $K_\varepsilon^\prime$, such that $m_k(\gamma) < \varepsilon, \, \forall k > K_\varepsilon^\prime$.
Let then choose an integer $N_\varepsilon$, such that, for all  $n \ge N_\varepsilon$, we have $K_\varepsilon = n p  - \sqrt{ \frac{ n p (1-p) }{ \varepsilon } } \ge K_\varepsilon^\prime$.
Then, letting for ease of notation $w_k = \binom{n}{k} p^k (1-p)^{n-k}$, for $k=0,1,\dots,n$, we have
\begin{align*}
   \frac{ U_n(p,\gamma) }{ \log_2(1+\gamma) } &= \sum_{ k = 0 }^{ n }{ m_k(\gamma) w_k }   \\
    &= \sum_{ k \le K_\varepsilon }{ m_k(\gamma) w_k } + \sum_{ k > K_\varepsilon }{ m_k(\gamma)w_k}   \\
    &\le M(\gamma) \mathcal{P}( B \le K_\varepsilon ) +  \varepsilon \mathcal{P}( B > K_\epsilon )  \\
    &\le [ M(\gamma) + 1 ]  \varepsilon
\end{align*}
for all $n > N_\varepsilon$.
This concludes the proof. \qed.

\section{Proof of scaling of $p_n$ and $\gamma_n$}
\label{app:scalingofpandgamma}
Ley us assume that $p_n \sim O(1/n^a)$ and $\gamma_n \sim O(1/n^b)$, with positive $a$ and $b$.

Let us first prove that $a+b > 1$ implies that $\lim_{ n \rightarrow \infty }{ U_n(p_n,\gamma_n ) } = 0$.
An obvious upper bound of $m_k(\gamma)$ for any $\gamma$ is $(1-\epsilon) k$.
Then, we have
\begin{equation}
U_n(p_n,\gamma_n) \le \log_2(1+\gamma_n) n p_n (1-\epsilon) \sim O\left( n^{1-a-b} \right)
\end{equation}
which proves that $U_n(p_n,\gamma_n) \rightarrow 0$ if $a+b > 1$.

Let us now assume $a+b < 1$.
Using the upper bound of $m_k(\gamma)$, we find
\begin{align*}
U_n(p,\gamma) &\le \log_2(1+\gamma) n^2 p  \left( 1 - \frac{ \gamma }{ 1+\gamma } p \right)^{n-1}   \\
  &= \log_2(1+\gamma) n^2 p \, e^{ (n-1) \log\left(  1 - \frac{ \gamma }{ 1+\gamma } p \right) }   \\
  &\sim O\left( n^{2-a-b} \, e^{ - n^{1-a-b} } \right)
\end{align*}
Thes proves that $U_n( p_n , \gamma_n ) \rightarrow 0$ as $n \rightarrow \infty$, in case $a+b < 1$.

\section{Proof of \cref{theo:introVj}}
\label{app:AppA}

We first characterize the order statistics of a set $\{ Y_1,\dots,Y_n \}$ of $n$ \ac{i.i.d.} negative exponential random variables with mean 1.
Let $Y_{(1)},Y_{(2)},\dots,Y_{(n)}$ denote the corresponding order statistics in descending order, i.e., we have $Y_{(1)} \ge Y_{(2)} \ge \dots \ge Y_{(n)}$.
The main difficulty with this set of random variables is that they are not independent.
It is possible however to find a simpler representation for the order statistics, namely, we can write:
\begin{equation}
\label{eq:orderstatrepresentation}
Y_{(k)} = \sum_{ j = k }^{ n }{ \frac{ X_j }{ j } }
\end{equation}
where $X_1,\dots,X_n$ are \ac{i.i.d.} negative exponential random variable with mean 1.

This representation is derived by considering $n$ negative exponential timers, starting off at the same time.
The first timer to expire marks the value of $Y_{(n)}$.
The residual time to expiry for all remaining timers is still distributed according to a negative exponential random variable, thanks to the memoryless property of the negative exponential random variable.
In general, let $T_k$ denote the time it takes for the next timer to go, when $k$ timers are left.
Thanks to the memoryless property, $T_k$ has the same probability distribution as the minimum of $k$ \ac{i.i.d.} negative exponential random variable of mean 1, i.e., $\mathcal{P}( T_k > t ) = e^{- k t }, \, t \ge 0$.
Let $t_0$ be the time when all timers are triggered.
The time $E_j$ when the $j$-th timer expires and only $j-1$ survive is $E_j = t_0 + T_n+\dots+T_{n-j+1}$.
This is the $j$-th smallest value among all timer lifetimes, hence $Y_{(n-j+1)} = E_j - t_0$ for $j = 1,\dots,n$ (remember the order statistics $\{ Y_{(k)} \}_{1 \le k \le n}$ that we consider is in \emph{descending} order).
We conclude then
\begin{equation}
Y_{(k)} = T_n+T_{n-1}+\dots+T_{k+1}+T_k
\end{equation}
where $T_j$ is a negative exponential random variable with mean $1/j$ and the $T_j$'s are independent.
Therefore, we can finally write the representation in \cref{eq:orderstatrepresentation}

Starting from \cref{eq:successfulSIC} and replacing $|\tilde{h}_k|^2$ with the descending order statistics $Y_{(k)}$ of a set of $n$ \ac{i.i.d.} negative exponential random variables with mean 1, we get the following condition for the successful reception of the packet associated with the $m$-th strongest signal:
\begin{equation}
	\label{eq:conditionset}
	Y_{(j)} \ge c + \gamma \sum_{ r = j+1 }^{ n }{ Y_{(r)} } + \gamma \xi \sum_{ r = 1 }^{ j-1 }{ Y_{(r)} } \, , \quad j = 1,\dots,m
\end{equation}

Using \cref{eq:orderstatrepresentation} into \cref{eq:conditionset}, we re-write the $j$-th condition above as follows:
\begin{equation}
\label{eq:SNIR3}
	\sum_{ k = j }^{ n }{ \frac{ X_k } { k } }  \ge c + \gamma \sum_{ r = j+1 }^{ n }{ \sum_{ k = r }^{ n }{ \frac{ X_k }{ k } } } + \gamma \xi \sum_{ r = 1 }^{ j-1 }{ \sum_{ k = r }^{ n }{ \frac{ X_k }{ k } } }
\end{equation}
Swapping the two last double summations, we get
\begin{equation}
	\sum_{ r = j+1 }^{ n }{ \sum_{ k = r }^{ n }{ \frac{ X_k }{ k } } } = \sum_{ k = j+1 }^{ n }{ \sum_{ r = j+1 }^{ k }{ \frac{ X_k }{ k } } } = \sum_{ k = j+1 }^{ n }{ X_k \, \frac{ k-j } { k } } 
\end{equation}
and
\begin{align}
   \sum_{ r = 1 }^{ j-1 }{ \sum_{ k = r }^{ n }{ \frac{ X_k }{ k } } } &= \sum_{ r = 1 }^{ j-1 }{ \left( \sum_{ k = r }^{ j-1 }{ \frac{ X_k }{ k } } + \sum_{ k = j }^{ n }{ \frac{ X_k }{ k } } \right) }  \\
   &= \sum_{ k = 1 }^{ j-1 }{ \sum_{ r = 1 }^{ k }{ \frac{ X_k }{ k } } } + \sum_{ k = j }^{ n }{ \sum_{ r = 1 }^{ j-1 }{ \frac{ X_k }{ k } } }  \\
   &= \sum_{ k = 1 }^{ j-1 }{ X_k } + \sum_{ k = j }^{ n }{ X_k \, \frac{ j-1 }{ k } } 
\end{align}

Plugging those expressions back into \cref{eq:SNIR3}, we find
\begin{equation}
\label{eq:SNIR4}
	\sum_{ k = j }^{ n }{ \frac{ X_k }{ k } }  \ge c + \gamma \sum_{ k = j }^{ n }{ X_k \, \frac{ k-j }{ k } } + \gamma \xi \sum_{ k = 1 }^{ j-1 }{ X_k } + \gamma \xi \sum_{ k = j }^{ n }{ X_k \, \frac{ j-1 }{ k } } 
\end{equation}
that is
\begin{equation}
\label{eq:SNIR5}
	\left[ 1 + j \gamma - (j-1) \gamma \xi \right] \sum_{ k = j }^{ n }{ \frac{ X_k }{ k } }  - \gamma \sum_{ k = j }^{ n }{ X_k } - \gamma \xi \sum_{ k = 1 }^{ j-1 }{ X_k } \ge c 
\end{equation}

Let us define the coefficients
\begin{equation}
\label{eq:bjkexpressionappendix}
b_{jk} = \begin{cases}
     - \gamma \xi & k = 1,\dots,j-1, \\
     \frac{ 1 + j \gamma - (j-1) \gamma \xi }{ k } - \gamma & k = j,\dots,n.
\end{cases}
\end{equation}
for $k = 1,\dots,n$.
Then, we can re-write \cref{eq:SNIR5} as follows
\begin{equation}
\label{eq:definitionofV}
V_j = \sum_{ k = 1 }^{ n }{ b_{jk} X_k } \ge c 
\end{equation}
where we have defined the random variable $V_j$.
The conditions for decoding the $m$ strongest packet successfully can be stated concisely as $V_j \ge c$ for $j =1,\dots,m$.

Since the mean and variance of $X_k$ is equal to 1 for all $k$ and the $X_k$'s are \ac{i.i.d.} random variables, we have
\begin{align}
    &\mu_j(n) = \mathrm{E}[ V_j ] = \sum_{ k = j }^{ n }{ b_{jk} }   \label{eq:muofjnappendix} \\
    &\sigma^2_j(n) = \mathrm{E}[ ( V_j - \mu_j )^2 ] = \sum_{ k = j }^{ n }{ b_{jk}^2 }   \label{eq:varianceofVjasasum}
\end{align}


\section{Proof of \cref{theo:asymptoticpropertiesofmomentsofVj}}
\label{app:AppB}

\subsection{Asymptotic regime for large $n$: mean value of $V_j$}
We recall that we are studying the asymptotic regime as $n \rightarrow \infty$ with $\gamma = 1/( \alpha n )$. 
The expression of the mean of $V_j$ can be derived explicitly from \cref{eq:bjkexpressionappendix,eq:muofjnappendix}:
\begin{equation}
\label{eq:meanofVj}
\mu_j(n) = A_j(n) \left[ H_n - H_{j-1} \right] - B_j(n)
\end{equation}
where
\begin{align}
    A_j(n) &= 1 + \frac{ \xi+(1-\xi) j }{ \alpha n }   \\
    B_j(n) &=   \frac{ n - (1-\xi)(j-1) }{ \alpha n }   \\
    H_n &= \sum_{ k = 1 }^{ n }{ \frac{ 1 }{ k } }
\end{align}
for $j = 1,\dots,n$.

The number $H_n$ is known as the $n$-th harmonic number.
For ease of notation, we let also $H_0 = 0$.
It is well known that $H_n \sim \log(n+1) + E$, where $E$ is the Euler-Mascheroni constant, $E \approx 0.5772$.
As for the difference of harmonic numbers, it is easy to verify that
\begin{equation}
\label{eq:harmonicnumberbounds}
\log\left( \frac{ n }{ j } \right) + \frac{ 1 }{ n } \le H_n - H_{j-1} = \sum_{ k = j }^{ n }{ \frac{ 1 }{ k } } \le \log\left( \frac{ n }{ j } \right) + \frac{ 1 }{ j }
\end{equation}
This is proved by starting from the inequalities $\frac{ 1 }{ k+1 } \le \frac{ 1 }{ x } \le \frac{ 1 }{ k }$, holding for $k \le x \le k+1$ and for any $k \ge 1$.
Then, integrating over $x$ and summing over $k$ leads to the inequalities in \cref{eq:harmonicnumberbounds}

Let us define the function $f_{\alpha,\xi}(x)$ for $x \in (0,1)$, given by
\begin{equation}
f_{\alpha,\xi}(x) = - \left( 1 + \frac{ (1-\xi) x }{ \alpha } \right) \log x - \frac{ 1 - (1-\xi) x }{ \alpha }
\end{equation}

After a long yet simple calculation, it is verified that
\begin{equation}
\label{eq:basicineq}
\frac{ 1 }{ n } + \epsilon_j(n) \le \mu_j(n) - f\left( \frac{ j }{ n } \right) \le  \frac{ 1 }{ j } + \delta_j(n)
\end{equation}
where
\begin{align}
    &\delta_j(n) = \frac{ \xi }{ \alpha n } \left( \log\left( \frac{ n }{ j } \right) + \frac{ 1 }{ j } \right)    \label{eq:deltadefinitions} \\
    &\epsilon_j(n) = \frac{ \xi }{ \alpha n^2 } - \frac{ (1-\xi)(1-j/n) }{ \alpha n } + \frac{ \xi }{ \alpha n } \log\left( \frac{ n }{ j } \right)   \label{eq:epsilondefinition}
\end{align}
for $j = 1,\dots,n$.
Since $\delta_j(n)$ for fixed $n$ is a monotonously decreasing function of $j$ and a monotonously increasing function of $\xi \in [0,1]$, for all $j = 1,\dots,n$ we have
\begin{equation}
\label{eq:boundondelta}
\delta_j(n) \le \xi \frac{ \log n + 1 }{ \alpha n } \le \frac{ \log n + 1 }{ \alpha n }   
\end{equation}

As for $\epsilon_j(n)$, it is a monotonously increasing function of $\xi \in [0,1]$.
Setting $\xi = 1$ in \cref{eq:epsilondefinition}, we get for all $j = 1,\dots,n$:
\begin{equation}
\label{eq:upperboundonepsilon}
\epsilon_j(n)  \le \frac{ \log(n/j) + 1/n }{ \alpha n } \le \frac{ \log n + 1/n }{ \alpha n }
\end{equation}
For $\xi = 0$, and for all $j = 1,\dots,n$, we have also:
\begin{equation}
\label{eq:lowerboundonepsilon}
\epsilon_j(n)  \ge  - \frac{ 1 - j/n }{ \alpha n } \ge - \frac{ 1 }{ \alpha n }
\end{equation}

Summing up, for any fixed $n$ and for all $j = 1,\dots,n$ and $\xi \in [0,1]$, we have
\begin{align}
    &0 \le \delta_j(n) \le \frac{ \log n + 1 }{ \alpha n }   \\
    &| \epsilon_j(n) | \le \max\left\{ \frac{ \log n + 1/n }{ \alpha n } \, , \frac{ 1 }{ \alpha n } \right\} \le \frac{ \log n + 1 }{ \alpha n }
\end{align}

Using the inequalities above in \cref{eq:basicineq}, we can bound the absolute deviation between the sequences of the mean values and of the approximating function:
\begin{align}
\Delta_j(n) &= \left| \mu_j(n) - f\left( \frac{ j }{ n } \right) \right|  \nonumber  \\
  &\le \max\left\{ \frac{ 1 }{ n } + | \epsilon_j(n) |, \frac{ 1 }{ j } + \delta_j(n) \right\}  \nonumber  \\
  &\le \max\left\{ \frac{ 1 }{ n } + \frac{ \log n + 1 }{ \alpha n }, \frac{ 1 }{ j } + \frac{ \log n + 1 }{ \alpha n } \right\}  \nonumber  \\
  &= \frac{ 1 }{ j } + \frac{ \log n + 1 }{ \alpha n }   \label{eq:boundonabsdiff}
\end{align}

We consider now the mean absolute error between the sequences $\{ \mu_j(n)\}_{1 \le j \le n}$ and $\{ f(j/n) \}|_{1 \le j \le n}$
In view of \cref{eq:boundonabsdiff}, we have
\begin{align}
\Delta(n) &= \frac{ 1 }{ n } \sum_{ j = 1 }^{ n }{ \left| \mu_j(n) - f\left( \frac{ j }{ n } \right) \right| }   \nonumber \\
  &\le \frac{ 1 }{ n } \sum_{ j = 1 }^{ n }{ \frac{ 1 }{ j } } + \frac{ \log n + 1 }{ \alpha n } \sim O\left( \frac{ \log n }{ n } \right) 
\end{align}
where the last (asymptotic) equality is a consequence of $\sum_{ j = 1 }^{ n }{ 1/j } = H_n \sim \log n$ for $n \rightarrow \infty$.
This proves that $\lim_{ n \rightarrow \infty }{ \Delta(n) } = 0$, i.e., the sequence of mean values $\mu_j(n)$, $j = 1,\dots,n$, approaches the continuous curve $f_{\alpha,\xi}(x)$, i.e., the points $\mu_j(n)$ move towards the curve $f_{\alpha,\xi}(x), \, x \in (0,1]$ as $n$ grows.
The result proved above determines also the rate of convergence, which appears to be rather slow.


\subsection{Asymptotic regime for large $n$: variance of $V_j$}
From \cref{eq:varianceofVjasasum}, the expression of the variance of $V_j$ is as follows:
\begin{equation}
\label{eq:varianceofVj}
\sigma^2_j(n) = a^2_j(n) \sum_{ k = j }^{ n }{ \frac{ 1 }{ k^2 } } - 2 \frac{ a_j(n) }{ \alpha n } \sum_{ k = j }^{ n }{ \frac{ 1 }{ k } } + \frac{ n - j + (j-1) \xi^2 }{ \alpha^2 n^2 }
\end{equation}
where
\begin{equation}
    a_j(n) = 1 + \frac{ \xi + (1-\xi) j }{ \alpha n }
\end{equation}

We distinguish two cases:
\begin{enumerate}
  \item Fixed $j$, so that the fraction $j/n$ becomes vanishingly small as $n$ grows.
  \item $j = n x$, with fixed $x$, so that the fraction $j/n$ is fixed and positive as $n$ grows.
\end{enumerate}

In the first case, for fixed $j$ we have for $n \rightarrow \infty$:
\begin{align*}
    &a_j(n) \rightarrow 1   \\
    &\sum_{ k = j }^{ n }{ \frac{ 1 }{ k^2 } } \rightarrow \frac{ \pi^2 }{ 6 } - \sum_{ k = 1 }^{ j-1 }{ \frac{ 1 }{ k^2 } }   \\
    &\sum_{ k = j }^{ n }{ \frac{ 1 }{ k } }  = H_n - H_{j-1} \sim \log\left( \frac{ n }{ j } \right)   \\
    &\frac{ n - j + (j-1) \xi^2 }{ \alpha^2 n^2 } \rightarrow 0
\end{align*}
where the second inequality stems from the known series sum $\sum_{ k = 1 }^{ \infty }{ 1/k^2 } = \pi^2/6$.
Using these limits in \cref{eq:varianceofVj}, we have for any given fixed $j$:

\begin{equation}
\lim_{ n \rightarrow \infty }{ \sigma^2_j(n) } = \frac{ \pi^2 }{ 6 } - \sum_{ k = 1 }^{ j-1 }{ \frac{ 1 }{ k^2 } }
\end{equation}

On the other hand, from \cref{eq:meanofVj}, it is apparent that for any given fixed $j$ we have $\lim_{ n \rightarrow \infty }{ \mu_j(n) } = \infty$.
Hence, for any given fixed $j$, we have
\begin{equation}
\lim_{ n \rightarrow \infty }{ \frac{ \sigma_j(n) }{ \mu_j(n) } } = 0
\end{equation}

In the second case, when the fraction $j/n = x$ is fixed as $n$ grows, we have

\begin{align*}
\sigma^2_j(n) &\le \left( 1 + \frac{ j }{ \alpha n } \right)^2 \sum_{ k = j }^{ n }{ \frac{ 1 }{ k^2 } } + \frac{ n - 1 }{ \alpha^2 n^2 }   \\
  &\le \left( 1 + \frac{ j }{ \alpha n } \right)^2 \frac{ n-j }{ j^2 } + \frac{ 1 }{ \alpha^2 n }    \\
  &= \left( 1 + \frac{ x }{ \alpha } \right)^2 \frac{ n - n x }{ n^2 x^2 } + \frac{ 1 }{ \alpha^2 n }   \\
  &= \left[  \frac{ 1-x }{ x^2 } \left( 1 + \frac{ x }{ \alpha } \right)^2 + \frac{ 1 }{ \alpha^2 } \right] \frac{ 1 }{ n }   
\end{align*}
The upper bound in the first line is obtained as a consequence of the following facts.
\begin{enumerate}
  \item The coefficients $a_j(n)$ are decreasing functions of $\xi \in [0,1]$ for all $j$, hence $a_j(n) \le a_j(n) |_{ \xi = 0 } = 1 + j/( \alpha n )$.
  \item The second term on the right hand side of \cref{eq:varianceofVj} is negative.
  \item The last term is a monotonously increasing function of $\xi$, hence it is maximized for $\xi = 1$.
\end{enumerate}
The expression in the third line is obtained by replacing $j$ with $n x$.

This proves that, for any fixed $x$, we have
\begin{equation}
\lim_{ n \rightarrow \infty }{ \sigma^2_j(n)|_{ j/n = x } } = 0
\end{equation}

As for the mean, substituting $j$ with $n x$ in \cref{eq:meanofVj}, and taking the limit for $n \rightarrow \infty$ with fixed $x > 0$, it is not difficult to verify that
\begin{equation}
\lim_{ n \rightarrow \infty }{ \mu_j(n)|_{ j/n = x } } = f_{\alpha,\xi}(x)
\end{equation}
where $f_{\alpha,\xi}(x)$ is defined in \cref{eq:definitionfofx}.

Putting together the limits of variance and mean of $V_j$ for $j/n = x$ for fixed $x > 0$, we get:
\begin{equation}
\lim_{ n \rightarrow \infty }{ \left.\frac{ \sigma_j(n) }{ \mu_j(n) } \right|_{ j/n = x } } = 0
\end{equation}

Applying Chebichev inequality, we can write
\begin{equation}
\mathcal{P}\left( \left| \frac{ V_j }{ \mu_j(n) } - 1 \right| > \epsilon \right) \le \frac{ \sigma^2_j(n) }{ \mu^2_j(n) \epsilon^2 }
\end{equation}
Taking the limit for $n \rightarrow \infty$, we have for any given $\epsilon > 0$:
\begin{equation}
\lim_{ n \rightarrow \infty }{ \mathcal{P}\left( \left| \frac{ V_j }{ \mu_j(n) } - 1 \right| > \epsilon \right) } = 0
\end{equation}
which proves that $V_j \sim \mu_j(n)$ almost surely for large $n$, i.e., $V_j$ tends to become a deterministic variable, equal to its mean, with its coefficient of variation (the ratio of the standard deviation to the mean) shrinking to 0.

\section{Proof of \cref{theo:mainresult}}
\label{app:AppC}
Note that $\lim_{ x \rightarrow 0+ }{ f_{\alpha,\xi}(x)  } = +\infty$ and $f_{\alpha,\xi}(1) = -\xi/\alpha \le 0$.
Hence, the set $\left\{ x : x \in (0,1) \text{ and } f_{\alpha,\xi}(x) = c \right\}$ is nonempty for any positive $c$.

As $n$ gets large, the sequence of random variables $V_j$ for $j = 1,\dots,n$ tend to a series of deterministic random variables equal to $\mu_j(n) \sim f_{\alpha,\xi}(j/n)$.

On the other hand, the condition for the packet carried by the $m$-th strongest signal to be decoded successfully is that $V_j \ge c$ for $j = 1,\dots,m$.
Asymptotically, this happens if $\mu_j(n) \ge c$, i.e., $f_{\alpha,\xi}(j/n) \ge c$ for $j = 1,\dots,m$.

The overall number of correctly decoded packets will be $m$ if and only if $f_{\alpha,\xi}(j/n) \ge c$, for $j = 1,\dots,m$, and $f_{\alpha,\xi}((m+1)/n) < c$.
Reminding the definition of $\zeta$, in the limit for $n \rightarrow \infty$, it is recognized that $m/n \sim \zeta$.
Hence, the average number of correctly decoded packets out of $n$ transmitted behaves as $m_n(\gamma_n) \sim \zeta n$ as $n$ tends to infinity.

\section*{Acknowledgment}
This work was partially supported by the European Union under the Italian National Recovery and Resilience Plan (NRRP) of Next Generation EU, partnership on ``Telecommunications of the Future'' (PE00000001 - program ``RESTART'').


\end{document}